\documentclass[12pt,draftcls,onecolumn]{IEEEtran}
\usepackage[ruled,boxed]{algorithm2e}

\usepackage{amsmath,graphicx,color,psfrag,mathrsfs,bm}
\usepackage{hyperref}
\usepackage{pstool}
\usepackage{longtable}          
\usepackage{amsthm}
\usepackage{amssymb}

\usepackage{multicol}
\newtheorem{thm}{Theorem}[section]
\newtheorem{kor}[thm]{Corollary}
\newtheorem{lem}[thm]{Lemma}
\newtheorem{prop}[thm]{Proposition}

\theoremstyle{definition}
\newtheorem{defn}{Definition}[section]
\theoremstyle{problem}

\theoremstyle{remark}
\newtheorem{rem}{Remark}[section]

\begin{document}
%

\title{Asymptotically Minimax Robust Hypothesis Testing \footnote{\noindent \mbox {\hspace{-4.2mm} \bf This manuscript has been superseded by two separate papers:}
\begin{itemize}
    \item[(i)] \emph{Asymptotically Minimax Robust Likelihood Ratio Test} (arXiv:2602.08174)
    \item[(ii)] \emph{From Asymptotic to Finite-Sample Minimax Robust Hypothesis Testing} (arXiv:2602.19803)
\end{itemize}
Readers are encouraged to consult these versions for the final results.}}

\author{%
  \IEEEauthorblockN{G\"okhan  G\"ul}\\
  \IEEEauthorblockA{Fraunhofer Institute for Microengineering and Microsystems (IMM)\\
                    Carl-Zeiss-Str. 18-20, 55129 Mainz, Germany\\
                    Email: Goekhan.Guel@imm.fraunhofer.de}
}

\maketitle

\begin{abstract}
The design of asymptotically minimax robust hypothesis testing is formalized for the Bayesian and Neyman-Pearson tests of \text{Type-I} and \text{Type-II}. The uncertainty classes based on the KL-divergence, $\alpha$-divergence, symmetrized $\alpha$-divergence, total variation distance, as well as the band model, moment classes and p-point classes are considered. Implications between single-sample-, all-sample- and asymptotic minimax robustness are derived. Existence and uniqueness of asymptotically minimax robust tests are proven using Sion's minimax theorem and the Karush-Kuhn-Tucker multipliers. The least favorable distributions and the corresponding robust likelihood ratio functions are derived in parametric forms, which can then be determined by solving a system of equations. The proposed theory proves that Dabak's design does not produce any asymptotically minimax robust test. Furthermore, it also generalizes the earlier works by Huber and Kassam by allowing analytical derivations, hence, providing answers to the questions 'how?', which were left unanswered. Simulations are provided to exemplify and evaluate the theoretical derivations.

\end{abstract}


%
\IEEEpeerreviewmaketitle

\section{Introduction}
In simple binary hypothesis testing, complete statistical knowledge of the data is required in order to be able to design optimum tests \cite{kay}. However, such an assumption is too strict and often does not hold in practice \cite{levy}. In such a case a reasonable approach is to consider composite hypothesis represented by a set or class of distributions. Both the parametric \cite{kassam,elsawy} as well as non-parametric approaches fall under this category \cite{nonparametric}. While parametric models, including those driven by M-estimators \cite{hube81_2}, implicitly assume that the shape of the distributions is still perfectly known, non-parametric approaches, for example the sign- or the Wilcoxon test make only mild assumptions on the set of underlying distributions, hence they are regarded as conservative approaches \cite{gulbook}.\\
Minimax robust hypothesis testing allows both parametric as well as non-parametric modeling of uncertainties by assuming that under the hypothesis $\mathcal{H}_j$ the true distribution $G_j$ of the received data belongs to an uncertainty class $\mathscr{G}_j$. The choice of the uncertainty classes is usually application dependent and most common choices are either model based, e.g. $\epsilon$-contamination model, or distance based, e.g., uncertainty classes based on the KL-divergence \cite{levy09}. Consequently, the ultimate goal of a designer is to find a decision rule $\hat\delta$ which minimizes a predefined objective function for the least favorable distributions (LFDs) $(\hat{G}_0, \hat{G}_1)\in \mathscr{G}_0\times \mathscr{G}_1$. Under some mild conditions, such a design provides the most powerful test in a well defined minimax sense, i.e. a robust test which provides the best guaranteeable detection performance irrespective of uncertainties imposed on the statistical model.\\
Existence of minimax robust tests are completely determined by the choice of uncertainty sets. In case a minimax robust test does not exist over deterministic decision rules, it may still exist over the set of randomized decision rules \cite{gul7}. The problem with such a design is that the designed test is minimax robust only for a single sample and cannot be extended to multiple samples while maintaining the minimax robustness \cite{gul6}. In the presence of multiple samples and absence of minimax robust tests, probably the best option is to consider asymptotically minimax robust tests, which minimize the asymptotic decrease rate of the error probability. In summary, minimax robust tests can be broadly classified into four categories in terms of the number of samples or the choice of uncertainty classes:
\begin{enumerate}
  \item \textbf{All-sample} minimax robust tests (over deterministic decision rules) \cite{hube65,hube73}.
  \item \textbf{Single-sample} minimax robust tests (over randomized decision rules) \cite{levy09,gul7,gul6}.
  \item \textbf{Asymptotically} minimax robust tests (over deterministic decision rules) \cite{dabak,moment}.
  \item \textbf{Uncertainty classes} defined on some probability space \cite{gulbook}.
\end{enumerate}
\subsection{Related work}
The earliest work in robust hypothesis testing is attributed to P. J. Huber, who published a robust version of the probability ratio test for the $\epsilon$-contamination and total variation classes of probability distributions in 1965 \cite{hube65}. Huber derived the least favorable distributions and showed that the clipped likelihood ratio test was the minimax robust test for both uncertainty classes. The conclusions of this work was later extended by Huber and Strassen to a larger class, which includes five different classes as special cases \cite{hube68}. The largest classes known, for which a minimax robust test exists and is a version of $\hat{l}=\hat{g}_1/\hat{g}_0$, are the $2$-alternating capacities \cite{hube73}, where $\hat{g}_j$ is the density function corresponding to $\hat{G}_j$. All aforementioned works \cite{hube65}, \cite{hube68} and \cite{hube73} are all-sample minimax robust, and it was shown that such tests do not always exist, for example when the uncertainty classes are built with respect to the KL-divergence \cite{gul5}.\\
Clipped likelihood ratio tests (CLRTs) resulting from the uncertainty classes in \cite{hube65} and \cite{hube68} are widely used in practice, especially to deal with outliers. However, the models leading to CLRTs may be unrealistic for many applications, e.g., values of data samples which tend to infinity and still have a positive probability are almost never seen, but such a scenario is fully considered by the models in \cite{hube65} and \cite{hube68}. This was first observed by Dabak and Johnson, who suggested that eliminating such distributions could lead to a smoothed uncertainty model, which may be better suited for practical applications, where modeling errors is of interest. Based on this idea, they considered the KL-divergence as the distance to build the uncertainty classes and derived the corresponding robust test for the asymptotic case, i.e. as the number of measurements tends to infinity \cite{dabak}. Under several assumptions, Levy showed that a single-sample minimax robust test could be designed for the same uncertainty model, if the error minimizing decision rules are allowed to be randomized. Considering a similar approach all the assumptions made by Levy were later removed \cite{gul6}. The shortcomings of the model with the KL-divergence is that both the distance as well as the a priori probabilities of the nominal test are not selectable \cite{gul6}. Replacing the KL-divergence with the $\alpha$-divergence these two final constraints were also removed in \cite{gul7}. Surprisingly, Dabak and Johnson's asymptotically robust test was different from Levy's minimax robust test, which was also different from the CLRT. Even more interestingly, for the whole $\alpha$-divergence neighborhood and for any a priori probabilities of the hypotheses, the corresponding minimax robust test was a censored likelihood ratio test, with a well defined randomization function \cite{gul7}.\\
All aforementioned designs do not allow incorporating approximately known positions, shapes or statistics of the actual probability distributions into the considered model. Therefore, several other uncertainty models have been proposed in the literature. One approach is that the uncertainty classes can fully be defined in terms of the statistics of the actual distributions, such as the moments \cite{moment}. Another approach is to consider the p-point classes, which allow designation of the desired amount of area to the non-overlapping sub-sets of the domain of density functions \cite{elsawy,vastola}. The band models, which was first proposed by Kassam \cite{kassamband} and later revisited by \mbox{Fau\ss} et. al. \cite{fauss}, on the other hand, enable the assignment of the approximate shape and location to the actual distributions.\\
All above-mentioned works in the field of robust hypothesis testing are theoretical. There are also application oriented works, for example \cite{martin}, where Huber's clipped likelihood ratio test is applied to robust detection of a known signal in nearly Gaussian noise. These results are later strengthened for a known signal in contaminated non-Gaussian noise \cite{kassam2}. Robust detection of stochastic signals for Gaussian signal and Gaussian mixture noise is also studied for small and large samples sizes \cite{martin2}. Beside Huber's uncertainty classes, moment classes have also been used in various applications, such as finance \cite{smith}, admission control \cite{brichet} and queueing theory \cite{johnson}. P-point classes have been used in robust detection \cite{elsawy,elsawy2}, rate-distortion \cite{sarkison} and robust smoothing problems \cite{cimini} whereas band models have been used in robust land mine detection \cite{pambudi}, robust distributed detection \cite{leonard}, and robust and sequential gait symmetry detection \cite{ann}. In addition to direct detection procedures, robust detection has also been realized by means of robust estimation in radar applications considering complex elliptically symmetric distributions (CES) to model the clutter distribution \cite{tyler,pascal}, see also \cite{ollila} for a survey of the new results and applications.
\subsection{Motivation}
\begin{enumerate}
  \item The derivations in \cite{dabak, dabak2}, which are later summarized in \cite{levy}, do not yield asymptotically minimax robust Neyman-Pearson (NP) tests. This requires derivations and analysis, which lead to minimax robustness.
  \item The theoretical designs for the asymptotic case consider NP-formulations by default, probably because they result in simpler solutions \cite{dabak,moment}. However, by Chernoff \cite{chernoff}, it is well known that the NP-tests have the worst error exponents. Therefore, it is necessary to obtain the asymptotically minimax robust tests for the fastest decay rate of the error probability.
\end{enumerate}

\subsection{Summary of the paper and its contributions}
In this paper, the design of asymptotically minimax robust binary hypothesis tests is studied for various uncertainty classes. The existence and uniqueness of minimax robust tests are analyzed in general. Considering the Karush-Kuhn-Tucker (KKT) approach, the least favorable distributions and the robust likelihood ratio functions (LRFs) are derived in parametric forms, which can be made explicit by solving a set of non-linear equations. In the sequel, the contributions of this paper together with their relation to prior works are summarized.
\begin{enumerate} 
  \item It is shown that single-sample minimax robustness, hence all-sample minimax robustness (see Proposition~\ref{prop1}) implies asymptotically minimax robustness (see Proposition~\ref{prop2}). The pillar of this work which makes most of other contributions possible is to show that any minimax robust test can be designed via solving
  \begin{equation*}
    \min_{u\in(0,1)} \max_{(G_0,G_1)\in\mathscr{G}_0\times \mathscr{G}_1}D_u(G_0,G_1),
  \end{equation*}
where
\begin{equation*}
D_u(G_0,G_1)=\int_{\Omega}{g_1}^u {g_0}^{1-u}d\mu
\end{equation*}
is denoted as the $u$-divergence. The corresponding test is single-sample minimax robust if there exists one. Otherwise the test is asymptotically minimax robust (see Section~\ref{sec5}).
\item For the KL-divergence neighborhood, the LFDs of the asymptotically minimax robust \text{\emph{NP-tests}} of \text{Type-I} and \text{Type-II} are obtained in parametric forms. The parameters of LFDs can be found by solving four non-linear coupled equations and the corresponding test is different from the ones derived in \cite{dabak,dabak2,levy} (see Theorem~\ref{theorem03np} and Remark~\ref{theorem03np}).
\item For uncertainty classes based on the KL-divergence, $\alpha$-divergence, symmetrized $\alpha$-divergence, total variation distance as well as the band- and $\epsilon$-contamination models, the LFDs of the (asymptotically) minimax robust \emph{rate minimizing} tests are obtained in parametric forms. For moment classes and p-point classes, the design of minimax robust tests is defined as a convex optimization problem (see Section~\ref{sec6}).
\item The derivations regarding the total variation neighborhood generalize the ones obtained earlier by Huber \cite{hube65} via allowing unequal robustness parameters $\epsilon_0\neq \epsilon_1$ to be chosen (see Theorem~\ref{theorem03}). Moreover, the analytical designs regarding the $\epsilon$-contamination model, total variation neighborhood and the band model explain the choice of distributions and necessary parameters made by Huber \cite{hube65} and Kassam \cite{kassamband}, which were originally heuristically designed, i.e. the LFDs were some trial versions, which in turn yielded minimax robust tests. It is shown that two special cases of the band model give rise to two different versions of the $\epsilon$-contamination model which accept the CLRT as the corresponding minimax robust test (see Theorem~\ref{theorem1} and Theorem~\ref{theorem2}). It is also proven that both $\epsilon$-contamination models are single-sample minimax robust (see \cite{hube65} and Theorem~\ref{theoremeps}) and their intersection yields the general version of the band model.
\end{enumerate}

\subsection{Outline of the paper}
The rest of the paper is organized as follows. In Section~\ref{sec2}, a brief overview of the fundamental concepts in minimax robust hypothesis testing is given. In Section~\ref{sec3}, single-sample and all-sample minimax robustness are defined. In Section~\ref{sec4}, asymptotically minimax robustness is introduced and its relation to single- and all-sample minimax robustness is explained. In Section~\ref{sec5}, the equations formulating asymptotic minimax robustness are derived, saddle value condition is characterized and the problem statement is made. In Section~\ref{sec6}, the least favorable distributions and asymptotically minimax robust tests are obtained for various uncertainty classes. In Section~\ref{sec8}, simulations are performed to evaluate and exemplify the theoretical derivations. Finally in Section~\ref{sec9}, the paper is concluded.

\subsection{Notations}
The following notations are applied throughout the paper. Upper case symbols are used for probability distributions and random variables, and the corresponding lower case symbols denote the density functions and observations, respectively. Boldface symbols are used for the sequence of random variables, sequence of observations or joint functions. The hypotheses $\mathcal{H}_0$ and $\mathcal{H}_1$ are associated with the nominal probability measures $F_0$ and $F_1$, whereas the corresponding actual distributions are denoted by $G_0$ and $G_1$. The sets of probability distributions are denoted by $\mathscr G_0$ and $\mathscr G_1$, whereas $\mathscr{M}$ is denoted as the set of all distribution functions on $\Omega$. Every probability measure, e.g. $P[\cdot]$, is associated with its distribution function $P(\cdot)$ i.e., $P(y)=P[Y\leq y]$ for the random variable (r.v.) $Y$ and the observation $y$. The notation $\hat{(\cdot)}$ indicates the least favorable distributions or densities, e.g., $\hat {G}_j\in \mathscr G_j$, or the most favorable decision rules e.g. $\hat\delta$. The expected value of a random variable $Y$ is denoted by $\mathbb{E}[Y]$. The notation for the convex function $f$ is different than the notation for the nominal density functions $f_0$ or $f_1$, or the density functions $f_{\mathcal{L}}$ or $f_{\mathcal{N}}$, which correspond to the Laplacian and Gaussian density functions, respectively. Similarly, the likelihood ratio function denoted by $l$ has nothing in common with $L_0$ or $L_1$, which correspond to the Lagrangians. The Lagrangian parameters $\mu_0$ and $\mu_1$ are also different from the notation used for the measure $\mu$. The symbols $s_0$ and $s_0(\lambda_1)$ denote the same functions, where the latter is just explicitly written in terms of the unknown parameter. $A$ and $A_k$ are defined as some sets belonging to the underlined sigma-algebra $\mathscr{A}$. The argument (value on the domain) of the subsequent operation is denoted by $\arg$. The letter $t$ always indicates a threshold.


\section{Fundamentals of Minimax Robust Hypothesis Testing}\label{sec2}
Let $(Y_k)_{k\geq 1}$ be a sequence of independent and identically distributed (i.i.d.) random variables (r.v.s), each taking values on a measurable space $(\Omega,\mathscr{A})$, where $\mathscr{A}$ is the Borel $\sigma$-algebra and $\Omega$ is a set. Furthermore, let $\mathscr{M}$ denote the set of all probability distribution functions defined on $\Omega$, and let $\mathscr{G}_0\subset \mathscr{M}$ and $\mathscr{G}_1\subset \mathscr{M}$ denote two distinct subsets of $\mathscr{M}$, each associated with the hypothesis $\mathcal{H}_0$ and $\mathcal{H}_1$. The distribution of $Y_k$ is not known exactly but belongs to the uncertainty set $\mathscr{G}_j$ under the hypothesis $\mathcal{H}_j$. The goal is to decide which of the following hypothesis is true
\begin{align}\label{eq4}
\mathcal{H}_0&: Y_k \sim G_0,\quad  G_0\in\mathscr{G}_0,\nonumber\\
\mathcal{H}_1&: Y_k \sim G_1,\quad  G_1\in\mathscr{G}_1.
\end{align}
Given a sequence of observations $\boldsymbol{y}=(y_1,\ldots,y_n)$ from either of the hypothesis, which corresponds to $\boldsymbol{Y}=(Y_1,\ldots,Y_n)$, a statistical test (or a decision rule) is a measurable function $\delta:\boldsymbol{Y}\mapsto\{0,1\}$, which accepts the hypothesis $\mathcal{H}_{\delta(\boldsymbol{y})}$ and rejects the other. Let $\Delta$ be the set of all $\delta$ on $\Omega$, and $P_0=P(\mathcal{H}_0)$ and $P_1=P(\mathcal{H}_1)$ be the a priori probabilities of the hypotheses. Furthermore, let $P_F(\delta,G_0)=G_0[\delta(\boldsymbol{Y})=1]$ and $P_M(\delta,G_1)=G_1[\delta(\boldsymbol{Y})=0]$ define the false alarm and the miss detection probabilities, respectively, and
\begin{equation*}
P_E(\delta,G_0,G_1)=P_0P_F(\delta,G_0)+P_1P_M(\delta,G_1)
\end{equation*}
define the overall error probability. Then, a solution to the optimization problem
\begin{equation}\label{eq10}
\min_{\delta\in\Delta}\sup_{(G_0,G_1)\in {\mathscr{G}}_0\times{\mathscr{G}}_1}P_E(\delta,G_0,G_1)=\sup_{(G_0,G_1)\in {\mathscr{G}}_0\times{\mathscr{G}}_1}\min_{\delta\in\Delta}P_E(\delta,G_0,G_1)
\end{equation}
is sought. The testing procedure and the distribution functions which solve \eqref{eq10} are called the minimax robust decision rule $\hat \delta$ and the least favorable distribution functions (LFDs), $\hat{G}_0$ and $\hat{G}_1$, respectively. The equality sign in \eqref{eq10}, which is $\geq$ in general, is not taken for granted and requires a careful analysis both on the objective function as well as on the uncertainty sets, for instance by using Sion's minimax theorem \cite{sion}. A solution to \eqref{eq10} then implies a saddle value,
\begin{equation}\label{eq11}
P_E(\hat{\delta},G_0,G_1)\leq P_E(\hat{\delta},\hat{G}_0,\hat{G}_1)\leq P_E(\delta,\hat{G}_0,\hat{G}_1).
\end{equation}
Given the densities $\hat{g}_0$ and $\hat{g}_1$ corresponding to the LFDs, $\hat{G}_0$ and $\hat{G}_1$, the minimizing test $\hat\delta$ is known to be the likelihood ratio test
\begin{equation}\label{eq115}
\hat{\boldsymbol{l}}(\boldsymbol{y})=\prod_{k=1}^{n}\hat{l}(y_k)\stackrel{\mathcal{H}_1}{\underset{\mathcal{H}_0}{\gtreqless}}t
\end{equation}
where $\hat{l}=\hat{g}_1/\hat{g}_1$ is the robust likelihood ratio function, $\hat{\boldsymbol{l}}$ is the joint robust likelihood ratio function and $t$ is a threshold. Depending on whether $n=1$, $n<\infty$ or as $n\rightarrow\infty$, analysis may differ. These cases and their interconnections will be handled in the next section. Uncertainty classes $\mathscr{G}_0$ and $\mathscr{G}_1$ will be made explicit whenever they are defined.

\section{Single-Sample and All-Sample Minimax Robustness}\label{sec3}
In the following single-sample minimax robustness is introduced. It is then extended to multiple samples. Existence and uniqueness of single-sample and finite-sample minimax robust tests are mentioned in this section and that of the asymptotic minimax robust test are presented more in details in the next sections.

\begin{prop}\label{prop001}
Let $n=1$ and $Y=Y_1$. If there exist LFDs, $\hat{G}_0\in\mathscr{G}_0$ and $\hat{G}_1\in\mathscr{G}_1$ such that
\begin{align}\label{eq13x}
&G_0\left[\hat{l}(Y)< t\right]\geq \hat{G}_0\left[\hat{l}(Y)< t\right],\nonumber\\
&G_1\left[\hat{l}(Y)< t\right]\leq \hat{G}_1\left[\hat{l}(Y)< t\right],
\end{align}
for all $t\in\mathbb{R}_{\geq 0}$ and all $(G_0,G_1)\in\mathscr{G}_0\times\mathscr{G}_1$, the corresponding test given by \eqref{eq115} solves \eqref{eq10} for the $P_E$ minimizing $t$. 
\end{prop}

\begin{proof}
The proof is straightforward and is omitted, cf. \cite[p. 1754]{hube65}.
\end{proof}

\begin{defn}[Single-sample minimax robustness]\label{def1}
Existence of a single-sample minimax robust test stated by Proposition~\ref{prop001} is called single-sample minimax robustness.
\end{defn}

\noindent There is a tight connection between minimizing a distance and finding a solution to \eqref{eq13x}. This will be stated with the following theorem.
\begin{thm}\label{thm1}
Let $G_0$ and $G_1$ be two probability distributions which are both absolutely continuous with respect to a common measure $\mu$ on $\Omega$. Then, for the $f$-divergence \cite{osterreicher1981} defined by
\begin{equation}\label{eq3}
D_f(G_0,G_1)=\int_{\Omega}f\left(\frac{g_0}{g_1}\right)g_1 d \mu,
\end{equation}
where $f:\mathbb{R}_{\geq 0}\to \mathbb{R}$ is a convex function such that $f(1)=0$, we have
\begin{equation}\label{eq13}
(\hat{G}_0,\hat{G}_1)\in\mathscr{G}_0\times\mathscr{G}_1 \mbox{ satisfies \eqref{eq13x}} \Longleftrightarrow (\hat{G}_0,\hat{G}_1)\in\mathscr{G}_0\times\mathscr{G}_1 \mbox{ minimizes $D_f$}
\end{equation}
over all $(G_0,G_1)\in\mathscr{G}_0\times\mathscr{G}_1$ and for all twice continuously differentiable strictly convex functions $f$, where the $f$-divergence can alternatively be written as
\begin{equation}\label{eq14}
D_f(G_0,G_1)=\int_{0}^{\infty}\left(\min(1,t)-\int_\Omega\min (g_0,tg_1)d\mu\right) d\mu_f(t)
\end{equation}
with
\begin{equation*}
\mu_f(a,b]=\partial^r f(b)-\partial^r f(a),\quad 0<a<b<\infty
\end{equation*}
where $\partial^r$ denotes the right derivative operator.
\end{thm}
\begin{IEEEproof}
The equivalence stated by \eqref{eq13} was proven in \cite[Section 6]{hube73}\footnote{In \cite{hube73} $\Omega$ is defined to be a complete separable metrizable space. Furthermore, if all $(G_0,G_1)\in\mathscr{G}_0\times\mathscr{G}_1$ are absolutely continuous with respect to a fixed measure $\mu$, $\Omega$ may need to be finite.} for a version of $D_f$;
\begin{equation*}
D_{f^{*}}(G_0,G_1)=\int_{\Omega}f^{*}\left(\frac{g_0}{g_0+g_1}\right)(g_0+g_1)d \mu,
\end{equation*}
where $f^{*}:[0,1]\to \mathbb{R}$ is a twice continuously differentiable and strictly convex function. Let $f^{**}(y)=f^{*}(y)-f^{*}(1/2)$ which results in $f^{**}(1/2)=0$. By \cite[Equation 7,8]{vajda}, it is known that $D_{f}=D_{f^{**}}$ using the transformation $f^{**}(t)=tf((1-t)/t)$ for $t\in(0,1)$. Hence, minimizing $D_{f}$ and $D_{f^{**}}$ are equivalent over all twice continuously differentiable and strictly convex $f$. An alternative definition of $D_f$ given by \eqref{eq14} can be found in \cite[Equation 10]{sharp}.
\end{IEEEproof}
Depending on the definition of the uncertainty classes, $\mathscr{G}_0$ and $\mathscr{G}_1$, a minimax robust test may or may not exist. For example it exists for the $\epsilon$-contamination classes of distributions \cite{hube65} and it does not exist for the classes of distributions based on the KL-divergence \cite{gul5}. If the uncertainty classes do not allow a minimax robust test to exist for all thresholds, it may still be possible to obtain a minimax robust test for a unique decision rule, if randomized decision rules are allowed, see \cite{gul5}. The information in the randomization is lost by the multiplication of the likelihood ratios therefore straightforward extension of the test to multiple observations is not minimax robust. However, if a single-sample minimax robust test exists, then it also implies existence of a finite-sample minimax robust test for any $n<\infty$. In order to prove this, let us first consider the following definition and lemma.

\begin{defn}\label{def2}
Let $X$ and $Y$ be two random variables taking values on the same measurable space $(\Omega,{\mathscr{A}})$, having cumulative distribution functions $G_0$ and $G_1$, respectively. $X$ is called stochastically larger than $Y$, i.e. $X\succeq Y$, if $G_1(x)\geq G_0(x)$ for all $x$.
\end{defn}

\begin{lem}\label{lem1}
Let $X_1$, $X_2$, $Y_1$ and $Y_2$ be four random variables on $(\Omega,{\mathscr{A}})$, out of which $X_1$ and $X_2$, and $Y_1$ and $Y_2$ are independent. If $X_1\succeq Y_1$ and $X_2\succeq Y_2$, then $X_1+X_2\succeq Y_1+Y_2$.
\end{lem}

\begin{IEEEproof}
Proof of Lemma~\eqref{lem1} is simple and is omitted.
\end{IEEEproof}

\begin{prop}\label{prop1}
If a single-sample minimax robust test exists, then a minimax robust test exists for any finite-sample size $n<\infty$ with the same LFDs and for the same uncertainty classes.
\end{prop}

\begin{IEEEproof}
If $\hat l(Y)$ satisfies the inequalities in \eqref{eq13x}, so does $X_k=\log \hat{l}(Y_k)$ for all $k$, but for all $t\in\mathbb{R}$. Application of Lemma~\ref{lem1} to the sum of $n$ r.v.s implies that $\sum_{k=1}^n X_k$ is stochastically larger under $\hat{G}_0$ than under $G_0$, and similarly under $G_1$ than under $\hat{G}_1$. Hence, by Definition~\ref{def2} we have
\begin{align}
G_0\left[\sum_{k=1}^n X_k\leq t\right] \geq \hat{G}_0\left[\sum_{k=1}^n X_k \leq t\right],\nonumber\\
G_1\left[\sum_{k=1}^n X_k\leq t\right] \leq \hat{G}_1\left[\sum_{k=1}^n X_k \leq t\right]
\end{align}
for all $t$ and all $(G_0,G_1)\in\mathscr{G}_0\times\mathscr{G}_1$. This proves the assertion since
\begin{equation}
\left\{\exp\left(\sum_{k=1}^n X_k\right)\leq \exp(t)\right\}\equiv\{\hat{\boldsymbol{l}}(\boldsymbol{Y})\leq t^{'}\}.
\end{equation}
\end{IEEEproof}

\begin{defn}[All-sample minimax robustness]\label{def1}
Existence of minimax robust tests for any finite-sample size $n<\infty$ stated by Proposition~\ref{prop1} is called all-sample minimax robustness.
\end{defn}

\begin{rem}\label{remark0}
Proposition~\ref{prop1} implies
\begin{equation*}
\text{Single-sample minimax robustness}\Longleftrightarrow\mbox{All-sample minimax robustness}
\end{equation*}
Notice that in fact $(Y_k)_{k\geq 1}$ are required to be only independent but not necessarily identically distributed. This allows different uncertainty classed to be considered for each $Y_k$.
\end{rem}
Existence of finite-sample minimax robust tests are established by Proposition~\ref{prop1}. Examples of such tests can be found in \cite{hube68, hube73}. However, as shown before, such tests may not always exist even for very simple examples. In this case a minimax robust test can be designed asymptotically as $n\rightarrow\infty$. Asymptotically minimax robust hypothesis testing will be introduced in the next section. Uniqueness of minimax robust finite-sample tests are postponed to the next sections, due to the ease of derivations using the asymptotic theory.

\section{Asymptotically Minimax Robustness}\label{sec4}
If the uncertainty model of interest does not allow all-sample minimax robustness, an asymptotic design is necessary. In the sequel, asymptotic minimax robustness will be formalized with first deriving the minimax equations, then, stating the existence of a saddle value and last, with a clear statement of the problem.

\begin{defn}[Asymptotically minimax robustness]
Let $S_n(\boldsymbol{X})=\frac{1}{n}\log \hat{\boldsymbol{l}}(\boldsymbol{X})$. Then, the tests satisfying
\begin{align*}
&\lim_{n\rightarrow\infty}\frac{1}{n}\log  G_0\left[S_n(\boldsymbol{X})> t\right]\leq \lim_{n\rightarrow\infty}\frac{1}{n}\log  \hat{G}_0\left[S_n(\boldsymbol{X})> t\right]\nonumber\\
&\lim_{n\rightarrow\infty}\frac{1}{n}\log  G_1\left[S_n(\boldsymbol{X})\leq t\right]\leq \lim_{n\rightarrow\infty}\frac{1}{n}\log  \hat{G}_1\left[S_n(\boldsymbol{X})\leq t\right]
\end{align*}
for a fixed $t$ and for all $(G_0,G_1)\in\mathscr{G}_0\times\mathscr{G}_1$ are called asymptotically minimax robust with respect to the chosen $t$.
\end{defn}

\subsection{Derivation of the Rate Functions}\label{sec41}

\begin{thm}\label{theorem0}
Let $\boldsymbol{Y}=(Y_1,\ldots,Y_n)$ be a sequence of i.i.d. r.v.s, where each $Y_k$ is distributed as $G_j\in \mathscr{G}_j$ under $\mathcal{H}_j$. Furthermore, let $\boldsymbol{X}=(X_1,\ldots,X_n)$ be the sequence of r.v.s induced by the robust log-likelihood ratios $X_k:=\log\hat{l}(Y_k)$, where each $X_k$ has a finite moment generating function
\begin{equation}\label{eq29}
M_{X_k}^j(u)=\mathbb{E}_{G_j}\left[\exp\left({u X_k}\right)\right]<\infty,\quad j\in\{0,1\}.
\end{equation}
For the test comparing ${S_n(\boldsymbol{X})=\frac{1}{n}\sum_{k=1}^n X_k}$ to a threshold $t$, if
\begin{equation}\label{eq26}
\mathbb{E}_{G_0}[X_k]<t<\mathbb{E}_{G_1}[X_k],
\end{equation}
then, as $n\to\infty$, false alarm and miss detection probabilities decrease exponentially
\begin{align}\label{eq27}
\lim_{n\rightarrow \infty}\frac{1}{n}\log G_0[S_n(\boldsymbol{X})> t]=-I_0(t)\\\label{eq275}
\lim_{n\rightarrow \infty}\frac{1}{n}\log G_1[S_n(\boldsymbol{X})\leq t]=-I_1(t)
\end{align}
with the rate functions given by
\begin{equation}\label{eq28}
I_j(t)=\sup_{u\in\mathbb{R}}\left(tu-\log M_{X_k}^j(u)\right),\quad j\in\{0,1\}.
\end{equation}
\end{thm}

\begin{IEEEproof}
The results (\eqref{eq27} and \eqref{eq275}, respectively) follow by applying Cram\'{e}r's theorem \cite{Cramer} twice, first under $\mathcal{H}_0$ to the sequence of r.v.s $(X_k)_{k\geq 1}$, each satisfying $\mathbb{E}_{G_0}[X_k]<t$, and second under $\mathcal{H}_1$ to the sequence of r.v.s $(-X_k)_{k\geq 1}$, each satisfying $t<\mathbb{E}_{G_1}[X_k]$. The details of the proof can be found in \cite[p. 84]{gulbook}, \cite[pp. 76-80]{levy}.
\end{IEEEproof}
\begin{kor}\label{corollary1}
Any of the following conditions is sufficient for \eqref{eq26} to hold.
\begin{enumerate}
\item There are $\hat{G}_0\in\mathscr{G}_0$ and $\hat{G}_1\in\mathscr{G}_1$ satisfying single-sample minimax robustness
\item $\mathbb{E}_{\hat{G}_1}[X_k]<\mathbb{E}_{G_1}[X_k]$ and $\mathbb{E}_{\hat{G}_0}[X_k]>\mathbb{E}_{G_0}[X_k]$ for all $(G_0,G_1)\in \mathscr{G}_0\times \mathscr{G}_1$
\item  $\min\mathbb{E}_{G_1}[X_k]>\max\mathbb{E}_{G_0}[X_k]$ for all $(G_0,G_1)\in \mathscr{G}_0\times \mathscr{G}_1$
\end{enumerate}
Here we have $1 \Longrightarrow 2$ and $2 \Longrightarrow 3$ and the converses of these implications are not true in general.
\end{kor}

\begin{rem}\label{remark1}
Even if none of the sufficient conditions listed above holds, it may still be possible to obtain an asymptotically minimax robust test. Let two disjoint subsets be $A_0\times A_1\subset \mathscr{G}_0\times \mathscr{G}_1$ and $B_0\times B_1\subset \mathscr{G}_0\times \mathscr{G}_1$, where $A_0\times A_1\cup B_0\times B_1=\mathscr{G}_0\times \mathscr{G}_1$. Suppose that for all $(G_0,G_1)\in A_0\times A_1$, \eqref{eq26} holds and for all $(G_0,G_1)\in B_0\times B_1$ we have $\mathbb{E}_{G_1}[X_k]<t<\mathbb{E}_{G_0}[X_k]$. Furthermore, suppose that either of the conditions
\begin{enumerate}
\item Whether $(G_0,G_1)$ belongs to $A_0\times A_1$ or $B_0\times B_1$ is known to the decision maker.
\item Increasing behavior of error probability is instantly detectable.
\end{enumerate}
is true. Then, flipping the binary decisions for all $(G_0,G_1)\in B_0\times B_1$ (or whenever increasing $P_E$ is detected) and keeping the original decisions for $(G_0,G_1)\in A_0\times A_1$ allow the convergence of error probability to zero asymptotically for all $(G_0,G_1)\in\mathscr{G}_0\times \mathscr{G}_1$. However, in this case two different robust tests may need to be designed. Notice that for the flipped decisions $I_0$ and $I_1$ are exchanged.
\end{rem}

\subsection{Optimum Threshold}\label{sec42}
For the selected robust likelihood ratio function $\hat{l}$, the asymptotic decrease rates of false alarm and miss detection probabilities can be found from \eqref{eq28} for any threshold $t$ satisfying \eqref{eq26}, cf. Corollary~\ref{corollary1} and Remark~\ref{remark1}. Of particular interest is the optimal value of $t$ which maximizes the asymptotic decrease rate of the error probability. This problem was first solved by Chernoff \cite{chernoff}. In the following, problem statement will be given together with a simple proof.\\
Consider the Bayesian error probability yielding from the likelihood ratio test stated in Theorem~\ref{theorem0},
\begin{equation*}
P_E(n,t)=P_0P_F(n,t)+(1-P_0)P_M(n,t),
\end{equation*}
where $n$ is the total number of samples and $t$ is the threshold. From \eqref{eq27} and \eqref{eq275} one can write
\begin{equation*}
P_E(n,t)\approx C_F(n)P_0\exp\left(-nI_0(t)\right)+C_M(n)(1-P_0)\exp\left(-nI_1(t)\right),
\end{equation*}
where $C_F$ and $C_M$ satisfy
\begin{equation}\label{eq32}
\lim_{n\rightarrow\infty}\frac{1}{n}\log C_F(n)=\lim_{n\rightarrow\infty}\frac{1}{n}\log C_M(n)=0.
\end{equation}
Hence, the exponential decay rate of the error probability is governed by $I_0$ and $I_1$.
\begin{thm}\label{theorem01}
The optimum threshold minimizing the error probability asymptotically is
\begin{equation}
\arg \min_{t}\lim_{n\to\infty}P_E(n,t)=0.
\end{equation}
\end{thm}
A proof of Theorem~\ref{theorem01} is given in Appendix~\ref{appendix3}.

\subsection{From Single-sample to Asymptotic Minimax Robustness}\label{sec43}
The connection between asymptotic and single-, hence, all-sample minimax robustness can be given with the following proposition.
\begin{prop}\label{prop2}
For any pair of uncertainty classes $(\mathscr{G}_0,\mathscr{G}_1)$,
\begin{equation*}
\text{Single-sample minimax robustness}\Longrightarrow\text{Asymptotically minimax robustness,}
\end{equation*}
and the converse is not true in general.
\end{prop}

\begin{IEEEproof}
There are two conditions which guarantee existence of asymptotically minimax robust tests, one of which is the Cram\'{e}r's inequalities and the other one is, as it will be shown in the next sections, the existence of LFDs maximizing $D_u$ for rate minimizing, or minimizing $D_{\mathrm{KL}}$ for Neyman-Pearson type tests. Single-sample minimax robustness implies LFDs which automatically satisfy Cram\'{e}r's inequalities as stated by Corollary~\ref{corollary1}. By Theorem~\ref{thm1}, single-sample minimax robustness also implies LFDs minimizing all $f$-divergences. Since $-D_u$ and $D_{\mathrm{KL}}$ are two special cases of the $f$-divergences, single-sample minimax robustness implies asymptotically minimax robustness.
\end{IEEEproof}

\begin{rem}\label{myrem}
Since asymptotically minimax robustness is a necessary condition for single-sample minimax robustness (and is unique as will be shown later), one can find the asymptotically minimax robust test, and claim that it is also single-sample minimax robust, hence finite-sample minimax robust, in case a single-sample minimax robust test exists.
\end{rem}
By considering Remark~\ref{myrem} it is now possible to derive single-sample minimax robust tests by using the asymptotic theory. As will be shown later, this is much easier than considering the theory from single-sample minimax robustness.

\section{Derivation of Minimax Equations and Problem Formulation}\label{sec5}
From the previous section the moment generating function $M_{X_k}^0(u)$ is equivalent to the \text{$u$-divergence} which is defined as
\begin{equation*}
D_u(G_0,G_1)=\int_{\Omega}{g_1}^u{g_0}^{1-u}d\mu.
\end{equation*}
In fact $D_u$ is a non-scaled version of the Renyi divergence \cite{renyi}. Some important properties of the \text{$u$-divergence} can be listed as follows
\begin{enumerate}
  \item $D_u$ is nonnegative, i.e. $D_u(G_0,G_1)\geq 0$
  \item $D_u$ is continuous and convex in $u$
  \item $D_u$ is continuous and jointly concave in $(g_0,g_1)$
  \item $D_{u=0}(G_0,G_1)=D_{u=1}(G_0,G_1)=1$
  \item $D_u(G_0,G_1)\in[0,1]$ from 1) and H\"older's inequality \cite{holder}
  \item $\hat u=\arg\min_u D_u\Longrightarrow \hat u\in[0,1]$ from 2) and 4)
  \item $D_u$ is related to the $\alpha$-divergence as $D_\alpha(G_0,G_1)=(1-D_u(G_0,G_1))/u(1-u)|_{u:=\alpha}$
\end{enumerate}
While $1)$ and $4)$-$7)$ are trivially correct, $2)$ and $3)$ follow from the properties of the $\alpha$-divergence using $7)$ \cite{entropy}.

\subsection{Minimax Equations}\label{sec5_3}
The minimax robust test that is intended to be designed is a likelihood ratio test with \text{$\hat l=\hat{g}_1/\hat{g}_0$}, for which the worst case data samples are also obtained from $\hat{g}_0$ and $\hat{g}_1$. As a result, $t=0$ can safely be selected as the optimum threshold, see the proof of Theorem~\ref{theorem01}. Hence, using $t=0$ together with \text{$-\inf(x)=\sup(-x)$} for $x<0$ and \text{$\inf(-x)=-\sup(x)$} the minimax equations can be obtained from \eqref{eq29} and \eqref{eq28} as
\begin{align}\label{eq37}
\hat{g}_0=&\arg\sup_{G_0\in \mathscr{G}_0}\left(\inf_{u_0\in \mathbb{R}}\log \int_{\Omega}\left(\frac{\hat{g}_1}{\hat{g}_0}\right)^{u_0} g_0 d\mu\right),\nonumber\\
\hat{g}_1=&\arg\sup_{G_1\in \mathscr{G}_1}\left(\inf_{u_1\in \mathbb{R}}\log \int_{\Omega}\left(\frac{\hat{g}_1}{\hat{g}_0}\right)^{u_1} g_1 d\mu\right).
\end{align}
In their current forms, these two coupled equations are mathematically intractable, especially if $\mathscr{G}_0$ and $\mathscr{G}_1$ are infinite sets. A solution found for \eqref{eq37} with the least favorable densities $\hat{g}_0$ and $\hat{g}_1$ implies
\begin{align}\label{eq38}
C_0(\hat{G}_0,\hat{G}_1;\hat{u}_0)=&\inf_{u_0\in \mathbb{R}}\log D_{u_0}(\hat{G}_0,\hat{G}_1) ,\nonumber\\
C_1(\hat{G}_0,\hat{G}_1;\hat{u}_1)=&\inf_{u_1\in \mathbb{R}}\log D_{1+u_1}(\hat{G}_0,\hat{G}_1).
\end{align}

\begin{prop}\label{prop3}
The results of both optimization problems in \eqref{eq38} are the same with $\hat{u}_0=-\hat{u}_1$, where $\hat{u}_0\in[0,1]$, i.e. $C_0(\hat{G}_0,\hat{G}_1;\hat{u}_0)=C_1(\hat{G}_0,\hat{G}_1;-\hat{u}_1)$.
\end{prop}

\begin{IEEEproof}
Using the $6$th property of $D_u$ in \eqref{eq38} implies a unique minimizing $\hat{u}_0$, which lies in $[0,1]$. Due to same reasoning, one can see that the minimizing $\hat{u}_1^*=-\hat{u}_1$ also lies in $[0,1]$. Since $-C_0$ is the Chernoff distance which is symmetric \cite[p. 82]{levy}, we have
\begin{equation*}
C_0(\hat{G}_0,\hat{G}_1;\hat{u}_0)=C_0(\hat{G}_1,\hat{G}_0;\hat{u}_0)=C_1(\hat{G}_0,\hat{G}_1;\hat{u}_1^*).
\end{equation*}
\end{IEEEproof}
Proposition~\ref{prop3} implies that both optimization problems are equivalent and have the same result for $\hat{u}_0=-\hat{u}_1$. Hence, it is sufficient to consider only one of them. Considering the first formulation the problem to be solved can be reduced to
\begin{equation}\label{eq42}
(\hat{g}_0,\hat{g}_1)=\arg \sup_{(G_0,G_1)\in \mathscr{G}_0\times \mathscr{G}_1} \inf_{u\in[0,1]}D_u(G_0,G_1)
\end{equation}
by removing the $\log$ term. This can be done because for both $\sup$ and $\inf$ we have
\begin{equation*}
\frac{\partial \log D_u}{\partial u}=0\Longrightarrow \frac{\partial D_u}{\partial u}=0
\end{equation*}
and $\log$ is an increasing mapping from $[0,1]$ to $\mathbb{R}_{\leq 0}$, see property $5)$ of $D_u$ .

\subsection{Saddle Value Condition}\label{sec5_4}
In this section existence of a saddle value, hence a solution to the minimax optimization problem in \eqref{eq42}, is discussed. Uniqueness condition depends on the choice of the uncertainty classes and will be discussed in the next section. In general existence of a saddle value is described by an existence of a solution to
\begin{equation}\label{eq45}
\min_{u\in[0,1]}\sup_{(G_0,G_1)\in \mathscr{G}_0\times \mathscr{G}_1} D_u(G_0,G_1)=\sup_{(G_0,G_1)\in \mathscr{G}_0\times \mathscr{G}_1} \min_{u\in[0,1]}D_u(G_0,G_1).
\end{equation}
By applying Sion's minimax theorem \cite{sion} it is shown in Appendix~\ref{appendix2} that there exists a saddle value to \eqref{eq45} and hence we have
\begin{equation}\label{eq46}
D_{\hat u}(G_0,G_1)\leq D_{\hat u}(\hat{G}_0,\hat{G}_1)\leq D_u(\hat{G}_0,\hat{G}_1),
\end{equation}
where $\hat u$ is the minimizing $u$, and $\hat{G}_0$ and $\hat{G}_1$ are the least favorable distributions.
\subsection{Problem Statement}\label{sec5_5}
From \eqref{eq46}, given $(\hat{G}_0,\hat{G}_1)$, the objective function $D_u$ needs to be minimized over $u$ and given the minimizing $\hat u$, $D_{\hat u}$ needs to be maximized over $(G_0,G_1)$. This can compactly be written as
\begin{equation}\label{eq47}
    \begin{aligned}[b]
        \text{Maximization:}\quad &
        \begin{aligned}[t]
            &\hat{g}_0=\mathrm{arg}\sup_{G_0\in\mathscr{G}_0}D_{u}(G_0,G_1)
             \quad \text{s.t. $g_0>0$, $\Upsilon(G_0)=\int_{\Omega}g_0\, d\mu=1$}\\
            &\hat{g}_1=\mathrm{arg}\sup_{G_1\in\mathscr{G}_1}D_{u}(G_0,G_1)
             \quad \text{s.t. $g_1>0$, $\Upsilon(G_1)=\int_{\Omega}g_1\,d\mu=1$}
        \end{aligned}
        \\[12pt]
        \text{Minimization:}\quad &\hat u=\mathrm{arg}\min_{u\in  [0,1]}D_{u}(\hat{G}_0,\hat{G}_1).
    \end{aligned}
\end{equation}
The maximization stage involves two separate constrained optimization problems, which are coupled. The following minimization problem can be solved once $\hat{g}_0$ and $\hat{g}_1$ are derived as functions of $u$.

\section{Least Favorable Distributions and Asymptotically Minimax Robust Tests}\label{sec6}
In this section LFDs and the asymptotically minimax robust tests are derived for various uncertainty classes considering the minimax optimization problem given by \eqref{eq47}. Additionally, the asymptotic NP-tests are also derived. Complete derivations are carried out for the Kullback-Leibler (KL)-divergence neighborhood, and similar steps are skipped for the sake of clarity when the same procedure is repeated for the $\alpha$- and the symmetrized $\alpha$-divergences. The derivations also include the uncertainty classes based on the total variation distance as well as the band model, moment classes, and p-point classes. A meta algorithm summarizing the asymptotic minimax testing process is given in Algorithm~\ref{algorithm1}.

\subsection{KL-divergence Neighborhood}\label{sec6_1}
Consider the uncertainty classes
\begin{equation}\label{eq49}
{\mathscr{G}}_j=\{G_j:D_{\mathrm{KL}}(G_j,F_j)\leq \epsilon_j\},\quad j\in\{0,1\},
\end{equation}
which are induced by the KL-divergence
\begin{equation*}
D_{\mathrm{KL}}(G_j,F_j)=\int_{\Omega}\log\left(g_j/f_j\right)g_j d\mu,
\end{equation*}
where $F_j$ is the nominal distribution under $\mathcal{H}_j$. The KL-divergence is considered as the classical information divergence \cite{inftheory} and used in earlier works to create the uncertainty classes \cite{dabak,levy09}. It is a smooth distance, hence suitable to deal with modeling errors \cite{gul6}.

\subsubsection{Rate Minimizing Tests}
Asymptotically minimax robust tests for the KL-divergence neighborhood can be stated with the following theorem.
\begin{thm}\label{theorem02}
For the uncertainty classes given by \eqref{eq49}, the LFDs
\begin{align*}
\hat{g}_0&=\exp{\left[\frac{-\lambda_0-\mu_0}{\lambda_0}\right]} \exp{\left[\frac{(1-u)(\hat{g}_1/\hat{g}_0)^{u}}{\lambda_0}\right]} f_0,\\
\hat{g}_1&=\exp{\left[\frac{-\lambda_1-\mu_1}{\lambda_1}\right]} \exp{\left[\frac{u (\hat{g}_1/\hat{g}_0)^{-1+u}}{\lambda_1}\right]} f_1,
\end{align*}
with the robust likelihood ratio function
\begin{equation}\label{eq56}
\frac{\hat{g}_1}{\hat{g}_0}=\exp{\left[\frac{-\mu_1}{\lambda_1}+\frac{\mu_0}{\lambda_0}\right]}\exp{\left[\frac{u(\hat{g}_1/\hat{g}_0)^{-1+u}}{\lambda_1}+\frac{(-1+u)(\hat{g}_1/\hat{g}_0)^u}{\lambda_0}\right]}l
\end{equation}
provide a unique solution to \eqref{eq47}. Moreover, given $u$, the Lagrangian parameters $\lambda_0$ and $\lambda_1$, hence $\mu_0$ and $\mu_1$, can be obtained by solving
\begin{align}\label{eq59x3}
\int_{\Omega} r_0 \log\left(r_0/s_0\right)f_0 d\mu/s_0=\epsilon_0, \nonumber\\
\int_{\Omega} r_1 \log\left(r_1/s_1\right)f_1 d\mu/s_1=\epsilon_1, \nonumber\\
\hat{g}_1/\hat{g}_0=(r_1 s_0)/(r_0 s_1),
\end{align}
where
\begin{align*}
s_0(\lambda_0)=\int_{\Omega}r_0(\lambda_0,\hat{g}_1/\hat{g}_0)f_0 d\mu=\int_{\Omega}\exp\left[\frac{(1-u)(\hat{g}_1/\hat{g}_0)^u}{\lambda_0}\right]f_0 d\mu=\exp\left[\frac{\lambda_0+\mu_0}{\lambda_0}\right],\nonumber\\
s_1(\lambda_1)=\int_{\Omega}r_1(\lambda_1,\hat{g}_1/\hat{g}_0)f_1 d\mu=\int_{\Omega}\exp\left[\frac{u(\hat{g}_1/\hat{g}_0)^{-1+u}}{\lambda_1}\right]f_1 d\mu=\exp\left[\frac{\lambda_1+\mu_1}{\lambda_1}\right].
\end{align*}
\end{thm}
\begin{IEEEproof}
Consider the Lagrangian
\begin{equation}\label{eq50}
L_0(g_0,g_1,\lambda_0,\mu_0)=D_{u}(G_0,G_1)+\lambda_0(\epsilon_0-D_{\mathrm{KL}}(G_0,F_0))+\mu_0(1-\Upsilon(G_0)),
\end{equation}
where $\lambda_0$ and $\mu_0$ are the KKT multipliers. A solution to \eqref{eq50} can uniquely be determined, in case all KKT conditions are met \cite[Chapter 5]{bertsekas2003convex}, because $L_0$ is a  strictly concave functional of $g_0$, as $\partial^2 L_0/\partial g_0^2<0$ for every $\lambda_0>0$. Writing \eqref{eq50} explicitly, it follows that
\begin{equation}\label{eq51}
L_0(g_0,g_1,\lambda_0,\mu_0)=\int_{\Omega}\left[\left(\frac{g_1}{g_0}\right)^u -\lambda_0\log\left(\frac{g_0}{f_0}\right)-\mu_0\right]g_0 d\mu+\lambda_0\epsilon_0+\mu_0.
\end{equation}
Imposing the stationarity condition of KKT multipliers and hereby taking the G$\hat{\mbox{a}}$teaux's derivative of Equation~\eqref{eq51} in the direction of $\psi_0$, yields
\begin{equation*}
\int_{\Omega}\left[(1-u)\left(\frac{g_1}{g_0}\right)^u-\lambda_0\log\left(\frac{g_0}{f_0}\right)-\lambda_0-\mu_0\right]\psi_0 d\mu,
\end{equation*}
which implies
\begin{equation}\label{eq53}
(1-u){g_0}^{-u}{g_1}^{u}-\lambda_0\log g_0=\lambda_0+\mu_0-\lambda_0\log f_0,
\end{equation}
since $\psi_0$ is an arbitrary function. Similarly, taking the G$\hat{\mbox{a}}$teaux's derivative of
\begin{equation}\label{eq50x}
L_1(g_0,g_1,\lambda_1,\mu_1)=D_{u}(G_0,G_1)+\lambda_1(\epsilon_1-D_{\mathrm{KL}}(G_1,F_1))+\mu_1(1-\Upsilon(G_1)),
\end{equation}
with respect to $g_1$ in the direction of $\psi_1$ leads to
\begin{equation}\label{eq54}
u{g_1}^{-1+u}{g_0}^{1-u}-\lambda_1\log g_1=\lambda_1+\mu_1-\lambda_1\log f_1.
\end{equation}
From \eqref{eq53} and \eqref{eq54} the least favorable densities can be obtained as a functional of the robust likelihood ratio function $\hat{g}_1/\hat{g}_0$ and the nominal densities $f_0$ and $f_1$ as given in Theorem~\ref{theorem02}. Since the second Lagrangian $L_1$ is also strictly concave for every $\lambda_1>0$, the parametric forms of the LFDs are obtained uniquely. Similarly, for every pair of distribution functions, $(G_0,G_1)\in \mathscr{G}_0\times \mathscr{G}_1$, $D_u$ is strictly convex in $u\in(0,1)$. Hence, for $\hat{g}_0$ and $\hat{g}_1$ the minimizing $u$ is unique as well. In order to obtain the parameters, given any choice of $u$, originally there are four non-linear equations
\begin{align}\label{59x1}
\Upsilon(\hat{G}_0(\lambda_0,\mu_0,\lambda_1,\mu_1))&=1 \nonumber\\
\Upsilon(\hat{G}_1(\lambda_0,\mu_0,\lambda_1,\mu_1))&=1 \nonumber\\
D_{\mathrm{KL}}(\hat{G}_0(\lambda_0,\mu_0,\lambda_1,\mu_1),F_0)&=\epsilon_0 \nonumber\\
D_{\mathrm{KL}}(\hat{G}_1(\lambda_0,\mu_0,\lambda_1,\mu_1),F_1)&=\epsilon_1
\end{align}
which need to be solved together with \eqref{eq56}. These five equations can be reduced to three without any loss of generality. From the first and second equations we have $s_0$ and $s_1$ as functionals of $r_0$ and $r_1$, respectively, as given in Theorem~\ref{theorem02}. Using $(r_0,s_0)$ and $(r_1,s_1)$ in the last two equations of \eqref{59x1} as well as in \eqref{eq56} the three equations given in \eqref{eq59x3} can be obtained.
\end{IEEEproof}

\begin{rem}
By Theorem~\ref{theorem02} the strategy followed is first to perform the maximization and determine the LFDs in parametric forms and then the minimization to determine the minimizing $u\in(0,1)$. Although $D_u$ is convex for any pair of known densities, for the pair of LFDs given in parametric forms as in Theorem~\ref{theorem02} $D_u$ is not necessarily convex, i.e. it is non-trivial to show this via proving or disproving the positivity of second derivative and counterexamples are readily available, see Section~\ref{sec8}. This indicates that a reasonable approach to determine the minimizing $u\in(0,1)$ is solving the equations given by \eqref{eq59x3} together with a fourth equation to minimize the variable $u$.
\end{rem}

\subsubsection{Neyman-Pearson Tests}
The asymptotic NP-tests are designed in such a way that one of the error exponents has the highest exponential decay rate, while the other (although not wanted) has the lowest. For Type-I NP-tests the threshold is chosen as $t_0=\lim_{\epsilon\rightarrow 0}\mathbb{E}_{G_0}[\log\hat{l}(Y_1)]+\epsilon$ such that $P_F$ is asymptotically guaranteed to get below any $\epsilon>0$, while $P_M$ has the highest decay rate. Similarly, for Type-II NP-tests the threshold is chosen as $t_1=\lim_{\epsilon\rightarrow 0}\mathbb{E}_{G_1}[\log\hat{l}(Y_1)]-\epsilon$ such that $P_M$ is asymptotically guaranteed to get below any $\epsilon>0$, while $P_F$ has the highest decay rate. Using the thresholds $t_0$ and $t_1$ in \eqref{eq28} and keeping in mind that $I_1(t)=I_0(t)-t$, one can obtain for the Type-I NP-test $I_0(t_0)=0$ and $I_1(t_0)=-t_0=D_{\mathrm{KL}}(G_0,G_1)$ and for the Type-II NP-test, $I_1(t_1)=0$ and $I_0(t_1)=t_1=D_{\mathrm{KL}}(G_1,G_0)$. Hence, the minimax problem formulation becomes
\begin{equation}\label{eq60}
    \begin{aligned}[b]
        \text{Type-I NP-test:}\quad &
        \begin{aligned}[t]
            &\min_{(G_0,G_1)\in\mathscr{G}_0\times \mathscr{G}_1}D_{\mathrm{KL}}(G_0,G_1)
             \quad \text{s.t. $g_0>0$, $g_1>0$, $\Upsilon(G_0)=1$, $\Upsilon(G_1)=1$}\\
        \end{aligned}\\
        \text{Type-II NP-test:}\quad &
               \begin{aligned}[t]
          &\min_{(G_0,G_1)\in\mathscr{G}_0\times \mathscr{G}_1}D_{\mathrm{KL}}(G_1,G_0)
             \quad \text{s.t. $g_0>0$, $g_1>0$, $\Upsilon(G_0)=1$, $\Upsilon(G_1)=1$}
        \end{aligned}
    \end{aligned}
\end{equation}
A solution to the Type-I NP-test formulation can be stated with the following theorem.
\begin{thm}\label{theorem03np}
The LFDs of the asymptotically minimax robust Type-I NP-test are given by
\begin{align}\label{eq63}
\hat{g}_0&=\left(\lambda_1+\frac{\lambda_1}{\lambda_0}\right)^{-\frac{1}{\lambda_0}}\exp{\left[-1-{\frac{1+\mu_0}{\lambda_0}}\right]}W\left[\frac{\lambda_0 \exp{\left[{\frac{-\lambda_1-\lambda_1\mu_0+\lambda_0\mu_1}{(1+\lambda_0)\lambda_1}}\right]}l^{-\frac{\lambda_0}{1+\lambda_0}}}{(1+\lambda_0)\lambda_1}\right]^{-\frac{1}{\lambda_0}}f_0\nonumber\\
\hat{g}_1&=\left(\lambda_1+\frac{\lambda_1}{\lambda_0}\right)^{-\frac{1+\lambda_0}{\lambda_0}}\exp{\left[-1-{\frac{1+\mu_0}{\lambda_0}}\right]}W\left[\frac{\lambda_0 \exp{\left[{\frac{-\lambda_1-\lambda_1\mu_0+\lambda_0\mu_1}{(1+\lambda_0)\lambda_1}}\right]}l^{-\frac{\lambda_0}{1+\lambda_0}}}{(1+\lambda_0)\lambda_1}\right]^{-\frac{1+\lambda_0}{\lambda_0}}f_0
\end{align}
where $W$ is the Lambert-W function.
\end{thm}
\begin{IEEEproof}
The solution can again be obtained by KKT multipliers. Considering the Lagrangians
\begin{align}\label{eq61}
L_0(g_0,g_1,\lambda_0,\mu_0)=D_{\mathrm{KL}}(G_0,G_1)+\lambda_0(D_{\mathrm{KL}}(G_0,F_0)-\epsilon_0)+\mu_0(\Upsilon(G_0)-1),\nonumber\\
L_1(g_0,g_1,\lambda_1,\mu_1)=D_{\mathrm{KL}}(G_0,G_1)+\lambda_1(D_{\mathrm{KL}}(G_1,F_1)-\epsilon_1)+\mu_1(\Upsilon(G_1)-1),
\end{align}
and following the same steps as before, one can get, respectively,
\begin{align}\label{eq62}
g_1&= \exp{\left[{1+\lambda_0+\mu_0}\right]}g_0^{1+\lambda_0}{f_0}^{-\lambda_0},\\\label{eq62x1}
g_0&= g_1\left(\mu_1+\lambda_1\left(1+\log\left(g_1/f_1\right)\right)\right).
\end{align}
Solving \eqref{eq62} and \eqref{eq62x1} for $g_0$ and $g_1$, respectively, the least favorable densities of the asymptotically minimax robust Type-I NP-test can be obtained as given in Theorem~\ref{theorem02}.
\end{IEEEproof}

\begin{rem}\label{remnp}
The Type-II minimax robust NP-test can similarly be obtained either by following the same procedure for the objective function $D_{\mathrm{KL}}(G_1,G_0)$ or by considering the same equations given by \eqref{eq63}. To accomplish the latter one, before the optimization $\epsilon_0$ and $f_0$ need to be interchanged by $\epsilon_1$ and $f_1$ respectively, and after obtaining the LFDs, $\hat{g}_0$ needs to be interchanged by $\hat{g}_1$, cf. \eqref{eq60}. The related parameters can be obtained directly by solving the equations in \eqref{59x1} for the LFDs \eqref{eq63}. As a side note both NP-tests are the limiting tests of the rate minimizing asymptotically minimax robust test as
\begin{equation*}
\max_{(G_0,G_1)\in\mathscr{G}_0\times \mathscr{G}_1} D_u(G_0,G_1)\equiv \min_{(G_0,G_1)\in\mathscr{G}_0\times \mathscr{G}_1}D_\alpha(G_0,G_1), \quad\forall u=\alpha\in (0,1)
\end{equation*}
and $D_\alpha(G_0,G_1)$ converges to $D_{\mathrm{KL}}(G_0,G_1)$ and $D_{\mathrm{KL}}(G_1,G_0)$, respectively, for $\alpha\rightarrow\ 1$ and $\alpha\rightarrow\ 0$ by using the $7th$ property of $D_u$, see also Section~\ref{sec_alpha}.
\end{rem}
Interestingly, the robust likelihood ratio function $\hat l$ is a nonlinear functional of $l$ through $W$. Moreover, compared to the rate minimizing asymptotically minimax robust tests the LFDs of their NP-counterparts are given only as a functional of the nominal distributions without coupling with $\hat{g}_1/\hat{g}_0$. This simplification, however, results in a complication of the closed form LFDs. Note that the problem formulation given by \eqref{eq60} differs from that of Dabak's formulation \cite{dabak,dabak2}, see also \cite[pp. 250-255]{levy}, i.e. Dabak's test is the result of a joint minimization of $I_0(t_1)$ over all $G_1\in\mathscr{G}_1$ and $I_1(t_0)$ over all $G_0\in\mathscr{G}_0$, hence it yields simpler analytic forms for the LFDs, but not an asymptotically minimax robust test, see Section~\ref{sec8}. However, Dabak's test is surprisingly asymptotically minimax robust for the expected number of samples of the sequential probability ratio test \cite{gulbook,gul6}.

\subsection{$\alpha-$divergence Neighborhood}\label{sec_alpha}
An alternative for the choice of the uncertainty classes is through considering the $\alpha-$divergence,
\begin{equation*}
D_\alpha(G_j,F_j):=\frac{1}{\alpha(1-\alpha)}\left(\int_{\Omega}((1-\alpha)f_j+\alpha g_j -g_j^\alpha f_j^{1-\alpha}) d \mu\right),\quad\alpha\in\Omega\backslash \{0,1\}
\end{equation*}
which is a special case of the $f$-divergence and includes various distances as special cases \cite{liese87},\cite[p.1537]{entropy}, e.g. $D_{\mathrm{KL}}$ as $\alpha\rightarrow 1$. The LFDs resulting from the $\alpha-$divergence neighborhood is stated with the following theorem.
\begin{thm}\label{theorem04}
The least favorable densities of the $\alpha-$divergence neighborhood can be given as
\begin{align}\label{eq67}
\hat{g}_0=&\left(\frac{1-\alpha}{\lambda_0}\left(\mu_0-(1-u)\left(\frac{\hat{g}_1}{\hat{g}_0}\right)^u\right)+1\right)^{\frac{1}{\alpha-1}}f_0,\nonumber\\
\hat{g}_1=&\left(\frac{1-\alpha}{\lambda_1}\left(\mu_1-u\left(\frac{\hat{g}_1}{\hat{g}_0}\right)^{-1+u}\right)+1\right)^{\frac{1}{\alpha-1}}f_1,
\end{align}
where
\begin{equation}\label{eq68}
\frac{\hat{g}_1}{\hat{g}_0}=\left(\frac{\frac{1-\alpha}{\lambda_1}\left(\mu_1-u\left(\frac{\hat{g}_1}{\hat{g}_0}\right)^{-1+u}\right)+1}{\frac{1-\alpha}{\lambda_0}\left(\mu_0-(1-u)\left(\frac{g_1}{g_0}\right)^u\right)+1}\right)^{\frac{1}{\alpha-1}}l.
\end{equation}
\end{thm}
\begin{IEEEproof}
The proof follows by using the same Lagrangian approach as before, i.e. by replacing $D_{\mathrm{KL}}$ with $D_\alpha$ in \eqref{eq50} and \eqref{eq50x} and performing the derivations.
\end{IEEEproof}
The parameters are obtained similarly by solving four non-linear equations coupled with \eqref{eq68}.
\subsubsection{Special Cases}
The LFDs in \eqref{eq67} can explicitly be written for some special choices of the parameters. For instance if $\alpha=1/2$ and $u=1/2$, the robust likelihood ratio function simplifies to
\begin{equation*}
\hat{l}=\sum_{k=0}^2 c_k l^{\frac{k}{2}} 
\end{equation*}
where
\begin{align*}
c_0=\frac{0.25{\lambda_0}^2}{4{\lambda_0}^2{\lambda_1}^2+{\lambda_0}^2{\mu_0}^2+4{\lambda_0}^2{\lambda_1}\mu_0},\nonumber\\
c_1=\frac{\lambda_0(2\lambda_0\lambda_1+\mu_0\lambda_1)}{4{\lambda_0}^2{\lambda_1}^2+{\lambda_0}^2{\mu_0}^2+4{\lambda_0}^2{\lambda_1}\mu_0},\nonumber\\
c_2=\frac{2\lambda_0\lambda_1+\mu_0\lambda_1}{4{\lambda_0}^2{\lambda_1}^2+{\lambda_0}^2{\mu_0}^2+4{\lambda_0}^2{\lambda_1}\mu_0}.
\end{align*}
\subsection{Symmetric $\alpha-$divergence Neighborhood}
The $\alpha-$divergence is not a symmetric distance in general, where $\alpha=1/2$ is an exception. A symmetrized version of the $\alpha-$divergence,
\begin{equation*}
D_\alpha^s(G_j,F_j)=\frac{1}{\alpha(1-\alpha)}\left(\int_{\Omega}((f_j^\alpha -g_j^\alpha)(f_j^{1-\alpha}-g_j^{1-\alpha})) d \mu\right),\quad\alpha\in\Omega\backslash \{0,1\}
\end{equation*}
can be obtained by
\begin{equation*}
D_\alpha^s(G_j,F_j)=D_\alpha(G_j,F_j)+D_\alpha(F_j,G_j).
\end{equation*}
Symmetric $\alpha-$divergence is also an $f$-divergence \cite{gulbook,sharp} including various other symmetric divergences such as the symmetric Chi-squared- or the symmetric KL-divergence \cite{entropy}. The LFDs resulting from the symmetric $\alpha-$divergence neighborhood is stated with the following theorem.
\begin{thm}
The LFDs of the symmetric $\alpha-$divergence neighborhood can be written in terms of two coupled equations
\begin{align}\label{eq74}
&\frac{\lambda_0}{1-\alpha}\left(\frac{\hat{g}_0}{f_0}\right)^{2\alpha-1}+\left((1-u)\left(\frac{\hat{g}_1}{\hat{g}_0}\right)^{u}-\frac{\lambda_0}{\alpha(1-\alpha)}-\mu_0\right)\left(\frac{\hat{g}_0}{f_0}\right)^{\alpha}+\frac{\lambda_0}{\alpha}=0,\\\label{eq74x}
&\frac{\lambda_1}{1-\alpha}\left(\frac{\hat{g}_1}{f_1}\right)^{2\alpha-1}+\left(u\left(\frac{\hat{g}_1}{\hat{g}_0}\right)^{u-1}-\frac{\lambda_1}{\alpha(1-\alpha)}-\mu_1\right)\left(\frac{\hat{g}_1}{f_1}\right)^{\alpha}+\frac{\lambda_1}{\alpha}=0.
\end{align}
\end{thm}
\begin{IEEEproof}
The proof follows by using the same Lagrangian procedure as before.
\end{IEEEproof}

In general, \eqref{eq74} and \eqref{eq74x} need to be solved jointly with four non-linear equations obtained from the Lagrangian constraints in order to determine the parameters. It is however possible to reduce the total number of equations to five if $\alpha$ is given. The idea is to solve \eqref{eq74} and \eqref{eq74x} such that $\hat{g}_0/f_0=h_0(\hat{g}_1/\hat{g}_0)$ and $\hat{g}_1/f_1=h_1(\hat{g}_1/\hat{g}_0)$, respectively. Hence, $\hat{g}_1/\hat{g}_0=(h_1/h_0)l$ is the coupling equation, where $h_0$ and $h_1$ are some functions.

\subsection{Total Variation Neighborhood}
The total variation neighborhood is defined as
\begin{equation*}
{\mathscr{G}}_j=\{G_j:D_{\mathrm{TV}}(G_j,F_j)\leq \epsilon_j\},\quad j\in\{0,1\},
\end{equation*}
where
\begin{equation*}
D_{\mathrm{TV}}(G_j,F_j)=\frac{1}{2}\int_{\Omega}|g_j-f_j| d\mu.
\end{equation*}
The LFDs and the corresponding minimax robust test for the uncertainty classes created by the total variation neighborhood were found earlier by Huber \cite{hube65}. However, the design approach is heuristic, many choices of the parameters and/or functions are unknown and the test is obtained under the assumption that the robustness parameters are equal $\epsilon_0=\epsilon_1$. Since asymptotic minimax robustness is a necessary condition for all-sample minimax robustness, the minimax robust test resulting from the total variation neighborhood can also be analytically derived following the same design procedure as before. The following theorem substantiate this claim.
\begin{thm}\label{theorem03}
For the total variation neighborhood, the robust LRF is given by
\begin{align}\label{eq80}
\frac{\hat{g}_1}{\hat{g}_0} &= \begin{cases} t_l, & l<(k_0 t_l)/k_1,  \\
\frac{k_1}{k_0}l, &(k_0 t_l)/k_1\leq l\leq (k_0 t_u)/k_1 \\
t_u, & l>(k_0 t_u)/k_1 \end{cases},
\end{align}
where $t_l$, $t_u$, $k_0$ and $k_1$ are some real numbers. Moreover, the LFDs can be chosen as
\begin{align*}
\hat{g}_0 &= \begin{cases} (f_0+f_1)/c_1, & l<(k_0 t_l)/k_1,  \\
k_0f_0, &(k_0 t_l)/k_1\leq l\leq (k_0 t_u)/k_1 \\
(f_0+f_1)/c_2, & l>(k_0 t_u)/k_1 \end{cases}
\end{align*}
and
\begin{align*}
\hat{g}_1 &= \begin{cases} t_l(f_0+f_1)/c_1, &l<(k_0 t_l)/k_1,  \\
k_1f_1, &(k_0 t_l)/k_1\leq l\leq (k_0 t_u)/k_1 \\
t_u(f_0+f_1)/c_2, & l>(k_0 t_u)/k_1 \end{cases}
\end{align*}
where
\begin{equation*}
c_1=\frac{1}{k_0}+\frac{t_l}{k_1}\quad \mbox{and}\quad c_2=\frac{1}{k_0}+\frac{t_u}{k_1}.
\end{equation*}
\end{thm}

\begin{IEEEproof}
Consider the Lagrangians,
\begin{align}\label{eq75}
L_0(g_0,g_1,\lambda_0,\mu_0)&=D_u(G_0,G_1)+\lambda_0 (D_{\mathrm{TV}}(G_0,F_0)-\epsilon_0)+\mu_0(\Upsilon(G_0)-1)),\nonumber\\
L_1(g_0,g_1,\lambda_1,\mu_1)&=D_u(G_0,G_1)+\lambda_1 (D_{\mathrm{TV}}(G_1,F_1)-\epsilon_1)+\mu_1(\Upsilon(G_1)-1)).
\end{align}
There are three cases of interest.

\subsubsection*{\text{\bf{Case 1}.} $g_j=k_jf_j:$}
\noindent Here, no derivatives are necessary and we simply have $g_0=k_0f_0$ and $g_1=k_1f_1$.
\subsubsection*{\text{\bf{Case 2}.} $g_j<k_jf_j:$}
\noindent Taking the G$\hat{\mbox{a}}$teaux's derivatives of the Lagrangians, $L_0$ and $L_1$, respectively, leads to
\begin{align}\label{eq77}
\int ((1-u)(g_1/g_0)^u+\mu_0-\lambda_0)\psi d \mu=0,\nonumber\\
\int (u(g_1/g_0)^{u-1}+\mu_1-\lambda_1)\psi d \mu=0.
\end{align}
\subsubsection*{\text{\bf{Case 3}.} $g_j>k_jf_j:$}
\noindent Similarly we get
\begin{align}\label{eq78}
\int ((1-u)(g_1/g_0)^u+\mu_0+\lambda_0)\psi d \mu=0,\nonumber\\
\int (u(g_1/g_0)^{u-1}+\mu_1+\lambda_1)\psi d \mu=0.
\end{align}
Since Case 2 and Case 3 cannot coexist, from these three cases, at most three different disjoint sets can be defined:
\begin{align}\label{eq785}
A_1&=\left\{y:g_0=k_0f_0,g_1>k_1f_1\right\}\equiv\left\{y:g_1=k_1f_1,g_0<k_0f_0\right\}\equiv\left\{y:g_1>k_1f_1,g_0<k_0f_0\right\},\nonumber\\
A_2&=\left\{y:g_0=k_0f_0,g_1=k_1f_1\right\}, \nonumber\\
A_3&=\left\{y:g_0=k_0f_0,g_1<k_1f_1\right\}\equiv \left\{y:g_1=k_1f_1,g_0>k_0f_0\right\}\equiv\left\{y:g_1<k_1f_1,g_0>k_0f_0\right\}.
\end{align}
Solving the equations from \eqref{eq77} and \eqref{eq78} together with the condition in Case 1, we get
\begin{align}\label{eq79}
\frac{\hat{g}_1}{\hat{g}_0} &= \begin{cases} \frac{u(-\lambda_0+\mu_0)}{(1-u)(\lambda_1+\mu_1)}, &A_1  \\
\frac{k_1}{k_0}\frac{f_1}{f_0}, &A_2  \\
\frac{u(-\lambda_0-\mu_0)}{(1-u)(\lambda_1-\mu_1)}, &A_3 \end{cases}
\end{align}
where
\begin{equation*}
u=\frac{\log\left(\frac{\lambda_0+\mu_0}{-\lambda_0+\mu_0}\right)}{\log\left(\frac{\lambda_0+\mu_0}{-\lambda_0+\mu_0}\right)+\log\left(\frac{-\lambda_1+\mu_1}{\lambda_1+\mu_1}\right)}.
\end{equation*}
The robust LRF given in Theorem~\ref{theorem03} is then immediate by using \eqref{eq785} in \eqref{eq79}.\\
Clearly, the minimax robust test must be unique. However, the Lagrangian approach considered imposes no constraints on the choice of the LFDs as long as the LFDs yield the robust likelihood ratio function given by \eqref{eq80}. As a result, the LFDs can be chosen as given in Theorem~\ref{theorem03}. In order $\hat{g}_0$ and $\hat{g}_1$ to be continuous, the limits from the left and right should agree for two meeting points of the piece-wise defined functions. This implies:
\begin{align*}
l=&\frac{k_0 t_l}{k_1}\Longrightarrow f_1=f_0 \frac{k_0 t_l}{k_1}\quad\mbox{and}\quad k_0f_0=\frac{f_0+f_1}{c_1}\Longrightarrow c_1=\frac{1}{k_0}+\frac{t_l}{k_1},\\
l=&\frac{k_0 t_u}{k_1}\Longrightarrow f_1=f_0 \frac{k_0 t_u}{k_1}\quad\mbox{and}\quad k_0f_0=\frac{f_0+f_1}{c_2}\Longrightarrow c_2=\frac{1}{k_0}+\frac{t_u}{k_1}.
\end{align*}
Hence, the proof is completed.
\end{IEEEproof}

\begin{rem}
In Theorem~\ref{theorem03}, $\hat{g}_0$ and $\hat{g}_1$ are obtained in four parameters ($t_l$, $t_u$, $k_0$ and $k_1$). The parameters can again be determined by imposing the four constraints defined by \eqref{eq75}, cf. \eqref{59x1}. These results generalize Huber's results allowing the robustness parameters to be chosen without the restriction of $\epsilon_0=\epsilon_1$ \cite{hube65}. Moreover, it is clear why the robust test is unique, the densities are not necessarily and the parameters $c_1$ are $c_2$ are chosen as such. Additionally, the robust likelihood ratio test is independent of $u$. This result is in line with Theorem~\ref{thm1}
\end{rem}

\subsection{Band Model}\label{sec6_band}
So far, the nominal distributions have been assumed to be known or could roughly be determined before constructing the uncertainty classes. In fact, depending on the application nominal distributions may also be unknown; for example only partial statistics of the data samples may be available \cite{moment,vastola} or the shape of the actual distributions may lie within a given band \cite{kassamband}. These cases will be studied here and in the following two sections.\\
The band model is given by the uncertainty classes
\begin{equation}\label{eq84}
\mathscr{G}_j=\left\{G_j\in\mathscr{M}: g_{j}^L \leq g_j  \leq g_{j}^U \right\}
\end{equation}
where $\mathscr{M}$ is the set of all distribution functions on $\Omega$, and $g_{j}^L$ and $g_{j}^U$ are non-negative lower and upper bounding functions such that $\mathscr{G}_0$ and $\mathscr{G}_1$ are nonempty sets. This implies
\begin{equation*}
\int_{\Omega} g_{j}^L d\mu \leq 1\leq \int_{\Omega} g_{j}^U d\mu,\quad j\in\{0,1\}.
\end{equation*}
Moreover, $g_{j}^L$ and $g_{j}^U$ should be chosen such that $g_0$ and $g_1$ are distinct density functions, if not $\mathscr{G}_0\cap\mathscr{G}_1\neq \emptyset$ and minimax hypothesis testing is not possible.
Theoretically, there are two main reasons to study band models. First, the band models are not equivalent to distance based uncertainty classes introduced so far. Because, distribution functions which are not absolutely continuous with respect to nominal distributions can belong to the band model, while this is not possible for the $f$-divergence based uncertainty classes. This result also includes the total variation distance since for any chosen total variation based uncertainty class, the band model should accept $g_j^L=0$ and $g_j^U=\infty$. Otherwise, there are density functions of type $\alpha g_j+(1-\alpha)\delta_x$, where $\delta_x$ is a dirac delta function at $y=x$, which belong to the total variation based uncertainty classes but not to the band model. On the other hand, choosing $g_j^L=0$ and $g_j^U=\infty$ defines the set of all density functions on $\Omega$, which is definitely not produced by the total variation distance unless $\epsilon_j$ are infinite.\\
The second reason to consider the band models is that the band models are in general capacity classes, however, whether they are two alternating has been unclear \cite{vastola}. Therefore, the theory introduced by Huber was not directly applied to band models \cite{hube73}. In fact, Huber has never defined the LFDs explicitly in \cite{hube73}, i.e. the LFDs are the distributions which maximize a version of the $f$-divergence over all distributions belonging to the related uncertainty classes.\\
Practically, the main motivation behind considering band models is that for some applications the density functions estimated from the training data are expressed as lying within a confidence interval and for these applications the band classes are the natural uncertainty model \cite{kassamband}.\\
The asymptotically minimax robust test and least favorable distributions arising from the band model can similarly be obtained as before. Consider the Lagrangians:
\begin{align*}
L_0(g_0,g_1,\lambda_0,\theta_0,\mu_0)=D_{u}(G_0,G_1)+\lambda_0(g_0-g_0^L)+\nu_0(g_0^U-g_0)+\mu_0(\Upsilon(G_0)-1),\nonumber\\
L_1(g_0,g_1,\lambda_1,\theta_1,\mu_1)=D_{u}(G_0,G_1)+\lambda_1(g_1-g_1^L)+\nu_1(g_1^U-g_1)+\mu_1(\Upsilon(G_1)-1),
\end{align*}
where $\mu_j$ are scalar, and $\lambda_j$ and $\nu_j$ are functional Langrangian multipliers. Taking the Gateaux derivatives of the Lagrangians, at the direction of unit area integrable functions $\psi_0$ and $\psi_1$, respectively, leads to
\begin{align}\label{eq87}
\frac{\partial L_0}{\partial g_0}=\int\left((1-u)\left(\frac{g_1}{g_0}\right)^u+\lambda_0-\nu_0+\mu_0\right)\psi_0 d\mu=0,\nonumber\\
\frac{\partial L_1}{\partial g_1}=\int\left(u\left(\frac{g_1}{g_0}\right)^{u-1}+\lambda_1-\nu_1+\mu_1\right)\psi_1 d\mu=0,
\end{align}
Three cases can separately be investigated.
\subsubsection*{\text{\bf{Case 1}.} $g_0^U=\infty$ and $g_1^U=\infty$ \text{(no upper bounding functions):}}
\noindent In this case, letting $g_{j}^L=(1-\epsilon_j)f_j$, the band model can equivalently be written as the lower $\epsilon$-contamination model
\begin{equation*}
\mathscr{G}_j^{\epsilon^{-}}=\left\{G_j: G_j=(1-\epsilon_j)F_j+\epsilon_j H, H\in\mathscr{M} \right\}
\end{equation*}
where $f_j$ are the nominal densities and $0\leq\epsilon_j<1$ \cite{kassamband}. Since we have $\nu_0=0$ and $\nu_1=0$ everywhere, and hence, no constraints regarding the upper bounding functions are in effect, there are four conditions regarding the Lagrangians
\begin{align*}
L_0:&\quad g_0=g_0^L\quad \mbox{on}\quad A_0\quad \mbox{and}\quad g_0>g_0^L\quad \mbox{on}\quad \Omega\backslash A_0,\nonumber\\
L_1:&\quad g_1=g_1^L\quad \mbox{on}\quad A_1\quad \mbox{and}\quad g_1>g_1^L\quad \mbox{on}\quad \Omega\backslash A_1.
\end{align*}
The integrals in \eqref{eq87} are defined for $g_0>g_0^L$ and $g_1>g_1^L$, respectively. Since $\nu_0=0$ and $\nu_1=0$ everywhere, with the assumption that $\lambda_j$ are constant functions, it is the case that
\begin{align}\label{eq89}
\frac{g_1}{g_0}=\frac{1}{k_2}\quad \mbox{on}\quad \bar{A}_0=\Omega\backslash A_0=\{y:g_0>g_0^L\},\nonumber\\
\frac{g_1}{g_0}=k_1\quad \mbox{on}\quad \bar{A}_1=\Omega\backslash A_1=\{y:g_1>g_1^L\},
\end{align}
where $k_1$ and $k_2$ are some positive constants.

\begin{thm}\label{theorem1}
From \eqref{eq89}, it follows that the LFDs and the corresponding likelihood ratio function are unique and given by
\begin{equation}\label{eq90}
\hat{g}_0=\begin{cases}
g_0^L, &  y\in A_0 \\
k_2g_1^L, &  y\in \bar{A}_0
\end{cases},\quad
\hat{g}_1=\begin{cases}
g_1^L, &  y\in A_1 \\
k_1g_0^L, &  y\in \bar{A}_1
\end{cases},
\end{equation}
and
\begin{equation*}
\frac{\hat{g}_1}{\hat{g}_0}=\begin{cases}
\frac{1}{k_2}, &  y\in \bar{A}_0\cap A_1 \\
\frac{g_1^L}{g_0^L}, &  y\in A_0\cap A_1\\
k_1, &  y\in A_0\cap \bar{A}_1
\end{cases}.
\end{equation*}
\end{thm}
\begin{IEEEproof}
The claim follows from the conditions:
\begin{enumerate}
\renewcommand\labelenumi{\bfseries\theenumi.}
\item The sets $A_0$, $A_1$, $\bar{A}_0$ and $\bar{A}_1$ are all non-empty.
\item The set $\bar{A}_0\cap \bar{A}_1$ is empty.
\item On $\bar{A}_0$ and $\bar{A}_1$, respectively, we have $\hat{g}_0=k_2g_1^L$ and $\hat{g}_1=k_1g_0^L$.
\end{enumerate}
\mbox{\bf 1.} The sets $\bar{A}_0$ and $\bar{A}_1$ are trivially non-empty. If not, we have $\int_{\Omega}\hat{g}_0=\int_{\Omega}g_0^L d\mu<1$ and $\int_{\Omega}\hat{g}_1 d\mu=\int_{\Omega}g_1^L d\mu<1$, which are contradictions with the fact that $\hat{g}_0$ and $\hat{g}_1$ are density functions. The set $A_0$ is also non-empty and this can be shown again with contradiction. Assume that $A_0$ is empty. In this case, $A_1$ can either be empty or non-empty. Assume that $A_1$ is also empty. Then, by \eqref{eq89}, we necessarily have $\hat{g}_0=\hat{g}_1$ a.e., which is excluded by a suitable choice of $g_0^L$ and $g_1^L$. Therefore, $A_1$ is non-empty. If $A_1$ is non-empty, then we must have $\hat{g}_0=k_2g_1^L$ on $\bar{A}_0$. If not, $\hat{g}_1/\hat{g}_0$ will not be a constant function on $\bar{A}_0\cap A_1$, which is non-empty since $\Omega=\bar{A}_0$. This again yields a contradiction with $\eqref{eq89}$. Since $\hat{g}_0=k_2g_1^L$ is defined on $\Omega$, in order to satisfy \eqref{eq89}, $\hat{g}_1$ must also be  $g_1^L$ on $\bar{A}_1$. Hence, we have $\hat{g}_1=g_1^L$ a.e. which is again a contradiction with the fact that $\int_{\Omega}\hat{g}_1=1$. Therefore, $A_0$ is non-empty. A similar analysis shows that $A_1$ is also non-empty.\\
\mbox{\bf 2.} The set $\bar{A}_0\cap \bar{A}_1$ is empty. If not, from \eqref{eq90} and \eqref{eq89} we have
\begin{equation}\label{eq93}
\frac{\hat{g}_1}{\hat{g}_0}=\frac{k_1}{k_2}\frac{g_0^L}{g_1^L}=k_1=\frac{1}{k_2}.
\end{equation}
This implies ($\bar{A}_0\cap A_1)\cup (A_0\cap \bar{A}_1)=\Omega$, hence, both $\bar{A}_0\cap \bar{A}_1$ and $A_0\cap A_1$ are empty sets. Since, $A_0\cap A_1$ is non-empty, we have a contradiction, hence, $\bar{A}_0\cap \bar{A}_1$ must be empty.\\
\mbox{\bf 3.} The set $A_0\cap A_1$ is non-empty. If not, $A_0$ and $A_1$ are disjoint sets. This implies at least non-empty $\bar{A}_0\cap A_1$ and $A_0\cap \bar{A}_1$ and at most additionally non-empty $\bar{A}_0\cap \bar{A}_1$. Non-empty $\bar{A}_0\cap \bar{A}_1$ implies $\hat{g}_1/\hat{g}_0=k$ a.e on $\Omega$, see \eqref{eq93}, and this is impossible, unless $k=1$. If only $\bar{A}_0\cap A_1$ and $A_0\cap \bar{A}_1$ are non-empty, i.e. if $\bar{A}_0\cap \bar{A}_1$ and $A_0\cap A_1$ are empty, hence, $(\bar{A}_0\cap A_1) \cup (A_0\cap \bar{A}_1)=\Omega$, we have $A_0= \bar{A}_1$ and $A_1=\bar{A}_0$ together with $A_0\cup A_1=\Omega$. This is possible if and only if $k=k_1=1/k_2$, because
\begin{align*}
\bar{A}_0\cap A_1&=\{g_0>g_0^L,g_1=g_1^L\}=\{1/k_2=g_1/g_0<g_1^L/g_0^L\}, \nonumber\\
A_0\cap \bar{A}_1&=\{g_1>g_1^L,g_0=g_0^L\}=\{k_1=g_1/g_0>g_1^L/g_0^L\}.
\end{align*}
The condition $k=k_1=1/k_2$ also implies $\hat{g}_1/\hat{g}_0=1$ a.e on $\Omega$, which is avoided by suitable choices of $g_0^L$ and $g_1^L$. Hence, $A_0\cap A_1$ cannot be empty.\\
The sets $\bar{A}_0\cap A_1$ and $A_0\cap \bar{A}_1$ are both non-empty. From $A_0\cap A_1\neq\emptyset$, there are four cases  $A_0\subset A_1$, $A_1\subset A_0$, $A_0=A_1$, or $A_0\backslash A_1$ and $A_1\backslash A_0$ are both non-empty.
The first three conditions imply either non-empty $\bar{A}_0\cap \bar{A}_1$, or $A_0=\Omega$, $A_1=\Omega$ or both. The first condition is a contradiction with \eqref{eq93} and the other three imply $\hat{g}_j=g_j^L$ on $\Omega$, which is impossible, see \eqref{eq90}. Therefore, we have non-empty $A_0\cap A_1$ together with non-empty $A_0\backslash A_1$ and $A_1\backslash A_0$. This eventually implies non-empty $\bar{A}_0\cap A_1$ and $A_0\cap \bar{A}_1$.\\
It is known that $\hat{g}_1=g_1^L$ on $A_1$ and on $\bar{A}_0\cap A_1$ we have $\hat{g}_1/\hat{g}_0=1/k_2$. Hence, on $\bar{A}_0$ we must have $\hat{g}_0=k_2 g_1^L$. Similarly, on $\bar{A}_1$ we have $\hat{g}_1=k_1 g_0^L$.\\
\end{IEEEproof}

\begin{kor}\label{corollary2}
The parameters should satisfy $k_1<1/k_2$, hence,
\begin{equation*}
A_0\cap A_1=\{k_1\leq g_1^L/g_0^L \leq 1/k_2\}.
\end{equation*}
Moreover,
\begin{equation*}
A_0=\{g_1^L/g_0^L<1/k_2\},\quad A_1=\{g_1^L/g_0^L>k_1\}.
\end{equation*}
\end{kor}

\begin{IEEEproof}
$k_1=1/k_2$ implies empty $A_0\cap A_1$, which is impossible, and $k_1>1/k_2$ implies non-empty $(\bar{A}_0\cap A_1)\cap (A_0\cap \bar{A}_1)$, which in turn implies $k_1=1/k_2$, another contradiction. Therefore, we have $k_1<1/k_2$. Accordingly, the sets $A_0$ and $A_1$ can be written as
\begin{align*}
A_0=&(A_0\cap A_1)\cup (A_0\cap \bar{A}_1)=\{k_1\leq g_1^L/g_0^L \leq 1/k_2\}\cup \{k_1>g_1^L/g_0^L\}=\{g_1^L/g_0^L\leq 1/k_2\},\nonumber\\
A_1=&(A_0\cap A_1)\cup (\bar{A}_0\cap A_1)=\{k_1\leq g_1^L/g_0^L \leq 1/k_2\}\cup \{1/k_2<g_1^L/g_0^L\}=\{g_1^L/g_0^L\geq k_1\}.
\end{align*}

\end{IEEEproof}

\begin{rem}
Let $t_u=1/k_2$, $t_l=k_1$ and $l=g_1^L/g_0^L$. Then, the LFDs and the robust LRF can be rewritten as
\begin{equation}\label{eq97}
\hat{g}_0=\begin{cases}
g_0^L, &  l\leq t_u \\
1/t_ug_1^L, &  l> t_u
\end{cases},\quad
\hat{g}_1=\begin{cases}
g_1^L, &  l\geq t_l \\
t_lg_0^L, &  l< t_l
\end{cases},
\end{equation}
and
\begin{equation}\label{eq98}
\frac{\hat{g}_1}{\hat{g}_0}=\begin{cases}
t_u, &  l>t_u \\
l, &  t_l\leq l \leq t_u\\
t_l, &  l<t_l
\end{cases}.
\end{equation}
The lower bounding function constraints are satisfied automatically. Because, on $\{l\leq t_u\}$ and $\{l\geq t_l\}$, $\hat{g}_j\geq g_j^L$ holds with equality, and on $\{l> t_u\}$ and $\{l< t_l\}$, we necessarily have $\hat{g}_0=1/t_ug_1^L\geq g_0^L$ and $\hat{g}_1= t_lg_0^L\geq g_1^L$, respectively, as $l=g_1^L/g_0^L$. The density function constraints are satisfied by solving
\begin{align}\label{eq99}
\int_{l\leq t_u} &g_0^L d\mu+\frac{1}{t_u}\int_{l> t_u} g_1^L \mathrm{d}\mu=1, \nonumber\\
\int_{l\geq t_l} &g_1^L d\mu+t_l\int_{l< t_l} g_0^L d\mu=1.
\end{align}
\end{rem}

\begin{rem}
As mentioned earlier the band model reduces to the $\epsilon$-contamination model if no upper bounding functions exist. In this case it is known by Huber that the equations in \eqref{eq99} have unique solutions and the LFDs in \eqref{eq97} are single-sample minimax robust \cite{hube65}. From Theorem~\ref{thm1}, single-sample minimax robust LFDs minimize all $f$-divergences, hence they also maximize all $u$-divergences. This proves that choosing $\lambda_j$ as scalars, which has been made to simplify the derivations, is a correct assumption. By \eqref{eq99}, it is also implied that the parameters $t_l$ and $t_u$ are only dependent on $g_0^L$ and $g_1^L$, i.e. they are independent of the choice of $u$. This is in accordance with Theorem~\ref{thm1}.
\end{rem}

\subsubsection*{\text{\bf{Case 2}.} $g_0^L=0$ and $g_1^L=0$ \text{(no lower bounding functions):}}
\noindent In this case, letting $g_{j}^U=(1+\epsilon_j)f_j$, the band model can equivalently be written as the upper $\epsilon$-contamination model
\begin{equation*}
\mathscr{G}_j^{\epsilon^{+}}=\left\{G_j: G_j=(1+\epsilon_j)F_j-\epsilon_j H, H\in\mathscr{M} \right\}
\end{equation*}
where $f_j$ are the nominal density functions and $\epsilon_j>0$. By the condition of no lower bounding functions, we have $\lambda_0=0$ and $\lambda_1=0$ everywhere. Similarly, the positivity constraints are also not imposed as before because, as it can be seen later, the density functions automatically satisfy these constraints. In this case, there are four conditions regarding the Lagrangians:
\begin{align}\label{eq100}
L_0:&\quad g_0=g_0^U\quad \mbox{on}\quad A_0\quad \mbox{and}\quad g_0<g_0^U\quad \mbox{on}\quad \Omega\backslash A_0,\nonumber\\
L_1:&\quad g_1=g_1^U\quad \mbox{on}\quad A_1\quad \mbox{and}\quad g_1<g_1^U\quad \mbox{on}\quad \Omega\backslash A_1.
\end{align}
The integrals in \eqref{eq87} are defined for $g_0<g_0^U$ and $g_1<g_1^U$, respectively. Since $\lambda_0=0$ and $\lambda_1=0$ everywhere, and with the assumption that $\nu_j$ are constant functions, it is the case that
\begin{align}\label{eq101}
\frac{g_1}{g_0}=\frac{1}{k_2}\quad \mbox{on}\quad \bar{A}_0=\Omega\backslash A_0=\{y:g_0<g_0^U\},\nonumber\\
\frac{g_1}{g_0}=k_1\quad \mbox{on}\quad \bar{A}_1=\Omega\backslash A_1=\{y:g_1<g_1^U\},
\end{align}
where $k_1$ and $k_2$ are some positive constants.

\begin{thm}\label{theorem2}
Let $t_l=1/k_2$, $t_u=k_1$ and $l=g_1^U/g_0^U$. It follows that the LFDs and the corresponding LRF are unique and given by
\begin{equation}\label{eq102}
\hat{g}_0=\begin{cases}
g_0^U, &  l\geq t_l  \\
1/t_lg_1^U, &  l< t_l
\end{cases},\quad
\hat{g}_1=\begin{cases}
g_1^U, &  l\leq t_u \\
t_ug_0^U, &  l> t_u
\end{cases},
\end{equation}
and
\begin{equation}\label{eq103}
\frac{\hat{g}_1}{\hat{g}_0}=\begin{cases}
t_l, &  l<t_l\\
l, &  t_l\leq l \leq t_u\\
t_u, &  l>t_u
\end{cases}.
\end{equation}
Moreover, all the Lagrangian constraints are satisfied and in particular the LFDs are obtained by solving
\begin{align*}
\int_{l\geq t_l} &g_0^U d\mu+\frac{1}{t_l}\int_{l< t_l} g_1^U \mathrm{d}\mu=1, \nonumber\\
\int_{l\leq t_u} &g_1^U d\mu+t_U\int_{l> t_u} g_0^U d\mu=1.
\end{align*}
\end{thm}

\begin{IEEEproof}
The definition of the sets $A_j$, their intersections, their relation to $l$, $k_1$ and $k_2$, and the fact that $k_1>1/k_2$ trivially follow from the same line of arguments used in Theorem~\ref{theorem1} and Corollary~\ref{corollary2} by considering \eqref{eq100} and \eqref{eq101}. The lower bounding function constraints are automatically satisfied as $\hat{g}_0$ and $\hat{g}_1$ are non-negative functions. The upper bounding function constraints are also satisfied in the same way as explained in Case $1$. The LFDs are obtained by unit density function constraints.
\end{IEEEproof}

\begin{thm}\label{theoremeps}
The LFDs in Theorem~\ref{theorem2} are single-sample minimax robust, i.e.
\begin{align}\label{eq105}
&G_0\left[\hat{l}< t\right]\geq \hat{G}_0\left[\hat{l}< t\right],\nonumber\\
&G_1\left[\hat{l}< t\right]\leq \hat{G}_1\left[\hat{l}< t\right]
\end{align}
for all $t\in\mathbb{R}_{\geq 0}$ and $(G_0,G_1)\in\mathscr{G}_0\times\mathscr{G}_1$, and/hence, $\nu_j$ can be chosen as constant functions.
\end{thm}

\begin{IEEEproof}
For any $g_j\in \mathscr{G}_j$, if $t>t_u$, the event $A=[\hat{l}< t]$ has a full probability and if $t\leq t_l$, it has a null probability. Therefore, \eqref{eq105} holds trivially for these cases. For $t_l<t\leq t_u$, we have
\begin{align*}
G_1(A)=&(1+\epsilon_1)F_1(A)-\epsilon_1 h \leq (1+\epsilon_1)F_1(A)=\hat{G}_1(A)\nonumber \\
G_0(A)=&(1+\epsilon_0)F_0(A)-\epsilon_0 h \geq (1+\epsilon_0)F_0(A)-\epsilon_0=1-(1+\epsilon_0)(1-F_0(A))\nonumber\\
      =&1-(1+\epsilon_0)F_0(\bar{A})=1-G_0^U(\bar{A})=1-\hat{G}_0(\bar{A})=\hat{G}_0(A).
\end{align*}
Hence, $\hat{g}_0$ and $\hat{g}_1$ are single-sample minimax robust. Moreover, by the virtue of Theorem~\ref{thm1}, single-sample minimax robust LFDs minimize all $f$-divergences, accordingly they also maximize all $u$-divergences. This proves that choosing $\nu_j$ as constant functions, which was made to simplify the derivations, was a correct assumption.
\end{IEEEproof}

\subsubsection*{\text{\bf{Case 3}.} $g_j^L<g_j<g_j^U$ \text{(the general case):}}
\noindent The uncertainty classes for the general case are obtained by the intersection of lower and upper $\epsilon$-contamination neighborhoods
\begin{equation*}
\mathscr{G}_j=\mathscr{G}_j^{\epsilon^{-}}\cap\mathscr{G}_j^{\epsilon^{+}}.
\end{equation*}
There are six conditions regarding the Lagrangians
\begin{align}\label{eq108}
L_0:&\quad g_0=g_0^L\quad \mbox{on}\quad A_0,\quad g_0=g_0^U\quad \mbox{on}\quad A_1\quad \mbox{and}\quad g_0^L<g_0<g_0^U \quad \mbox{on}\quad A_2,\nonumber\\
L_1:&\quad g_1=g_1^L\quad \mbox{on}\quad A_3,\quad  g_1=g_1^U\quad \mbox{on}\quad A_4\quad \mbox{and}\quad g_1^L<g_1<g_1^U\quad \mbox{on}\quad A_5.
\end{align}
The integrals in \eqref{eq87} are defined for $g_0^L<g_0<g_0^U$ and $g_1^L<g_1<g_1^U$, respectively. With the assumption that both $\lambda_j$ and $\nu_j$ are constant functions, it is the case that
\begin{align}\label{eq109}
\frac{g_1}{g_0}=k_2\quad \mbox{on}\quad  A_2=\{y:g_0^L<g_0<g_0^U\},\nonumber\\
\frac{g_1}{g_0}=k_1\quad \mbox{on}\quad  A_5=\{y:g_1^L<g_1<g_1^U\},
\end{align}
where $k_1$ and $k_2$ are some positive constants.

\begin{thm}\label{theorem3}
Assume that both $\lambda_j$ and $\nu_j$ are constant functions. Then, there are three different asymptotically minimax robust LRFs,
\begin{equation*}
\text{\rm Type-I}:\quad\frac{\hat{g}_1}{\hat{g}_0}=\begin{cases}
g_1^U/g_0^L, &  g_1^U/g_0^L\leq k_2\\
k_2, &  g_1^U/g_0^L>k_2>g_1^U/g_0^U\\
g_1^U/g_0^U, &  k_2\leq g_1^U/g_0^U\leq k_1\\
k_1, &  g_1^U/g_0^U>k_1>g_1^L/g_0^U\\
g_1^L/g_0^U, &  g_1^L/g_0^U\geq k_1
\end{cases},
\end{equation*}

\begin{equation*}
\text{\rm Type-II}:\quad\frac{\hat{g}_1}{\hat{g}_0}=\begin{cases}
g_1^U/g_0^L, &  g_1^U/g_0^L\leq k_1\\
k_1, &  g_1^U/g_0^L>k_1>g_1^L/g_0^U\\
g_1^L/g_0^U, &  g_1^L/g_0^U\geq k_1
\end{cases},
\end{equation*}

\begin{equation*}
\text{\rm Type-III}:\quad\frac{\hat{g}_1}{\hat{g}_0}=\begin{cases}
g_1^U/g_0^L, &  g_1^U/g_0^L\leq k_1\\
k_1, &  g_1^U/g_0^L>k_1>g_1^L/g_0^L\\
g_1^L/g_0^L, &  k_1\leq g_1^L/g_0^L\leq k_2\\
k_2, &  g_1^L/g_0^L>k_2>g_1^L/g_0^U\\
g_1^L/g_0^U, &  g_1^L/g_0^U\geq k_2
\end{cases},
\end{equation*}

with the corresponding pairs of LFDs, respectively,

\begin{equation*}
\hat{g}_0=\begin{cases}
g_0^L, &  g_1^U/g_0^L\leq k_2  \\
\frac{1}{k_2}g_1^U, &  g_1^U/g_0^L>k_2>g_1^U/g_0^U\\
g_0^U, &  g_1^U/g_0^U\geq k_2
\end{cases},\quad
\hat{g}_1=\begin{cases}
g_1^L, &  g_1^L/g_0^U\geq k_1 \\
k_1g_0^U, &  g_1^U/g_0^U>k_1>g_1^L/g_0^U\\
g_1^U, & g_1^U/g_0^U\leq k_1
\end{cases},
\end{equation*}

\begin{equation*}
\hat{g}_0=\begin{cases}
g_0^L, &  g_1^U/g_0^L\leq k_1  \\
k_2(g_0^L+h_1), &  g_1^U/g_0^L>k_1\geq g_1^L/g_0^L\\
\frac{k_2}{k_1}(g_1^L+h_2), &  g_1^L/g_0^L>k_1>g_1^L/g_0^U\\
g_0^U, &  g_1^L/g_0^U\geq k_1
\end{cases},\quad
\hat{g}_1=\begin{cases}
g_1^U, &  g_1^U/g_0^L\leq k_1 \\
k_1k_2(g_0^L+h_1), &  g_1^U/g_0^L>k_1\geq g_1^L/g_0^L\\
k_2(g_1^L+h_2), &  g_1^L/g_0^L>k_1>g_1^L/g_0^U\\
g_1^L, & g_1^L/g_0^U\geq k_1
\end{cases},
\end{equation*}

\begin{equation*}
\hat{g}_0=\begin{cases}
g_0^L, &  g_1^L/g_0^L\leq k_2 \\
\frac{1}{k_2} g_1^L, &  g_1^L/g_0^L>k_2>g_1^L/g_0^U\\
g_0^U, &  g_1^L/g_0^U\geq k_2
\end{cases},\quad
\hat{g}_1=\begin{cases}
g_1^L, &  g_1^L/g_0^L\geq k_1 \\
k_1 g_0^L, &  g_1^U/g_0^L>k_1>g_1^L/g_0^L\\
g_1^U, & g_1^U/g_0^L\leq k_1
\end{cases},
\end{equation*}
Moreover, LRFs of \text{Type-I} and \text{Type-III} tend to clipped likelihood ratio functions, e.g., as given by \eqref{eq98} with the corresponding LFDs defined by \eqref{eq97}.
\end{thm}

\begin{IEEEproof}
From \eqref{eq108} and \eqref{eq109}, LFDs can be written as
\begin{equation}\label{eq116}
\hat{g}_0=\begin{cases}
g_0^L, &  A_0 \\
\frac{1}{k_2} g_1^L\,\,\mbox{or } \frac{1}{k_2} g_1^U,&  A_2\\
g_0^U, &  A_1
\end{cases},\quad
\hat{g}_1=\begin{cases}
g_1^L, &  A_3 \\
k_1 g_1^L\,\,\mbox{or } k_1 g_1^U, &  A_5\\
g_1^U, &  A_4
\end{cases},
\end{equation}
Let $\hat{g}_0=\frac{1}{k_2} g_1^L$ on $ A_2$ and $\hat{g}_1=k_1 g_1^L$ on $A_5$. Then,
\begin{equation}\label{eq117}
A_1\cap  A_5,\quad A_2\cap  A_4,\quad \mbox{and}\quad A_2\cap  A_5
\end{equation}
are all empty sets, because their existence contradicts with \eqref{eq109}. Accordingly, the robust LRF can implicitly be written as
\begin{equation*}
\quad\frac{\hat{g}_1}{\hat{g}_0}=\begin{cases}
g_1^U/g_0^L, & A_0\cap  A_4 \\
k_1, & A_0\cap  A_5 \\
g_1^L/g_0^L, & A_0\cap  A_3 \\
g_1^U/g_0^U, & A_1\cap  A_4 \\
k_2, & A_2\cap  A_3 \\
g_1^L/g_0^U, &  A_1\cap  A_3
\end{cases}.
\end{equation*}
Furthermore, from \eqref{eq108} and \eqref{eq109} we have
\begin{align}\label{eq119}
A_0\cap A_5&=\{g_1^L<g_1<g_1^U,g_0=g_0^L\}=\{g_1^L/g_0^L<k_1=g_1/g_0<g_1^U/g_0^L\}, \nonumber\\
A_2\cap A_3&=\{g_0^L<g_0<g_0^U,g_1=g_1^L\}=\{g_1^L/g_0^U<k_2=g_1/g_0<g_1^L/g_0^L\}.
\end{align}
The empty sets in \eqref{eq117} imply $A_2\subset A_3$ and $A_5\subset A_0$, which in turn imply $A_5=A_0\cap A_5$ and $A_2=A_2\cap A_3$. Accordingly, $A_2$ and $A_5$ can also be made explicit in \eqref{eq116}. The sets $A_0$, $A_1$ and $A_2$ are disjoint, as well as the sets $A_3$, $A_4$ and $A_5$. On $A_2$ we have $g_1^L/k_2<g_0^U$ and due to continuity $\frac{1}{k_2} g_1^L=g_0^U$ at least on a single point. It is also at most on a single point, if not $A_1$ and $A_2$ are not disjoint. For $A_1$, the only choice left is then $A_1=\{g_1^L/k_2\geq g_0^U\}$. Similarly, i.e. considering $g_0^L< g_1^L/k_2$ on $A_2$ etc., we have $A_0=\{g_0^L\leq g_1^L/k_2\}$. Performing the same analysis over $A_2\cap A_3$, leads to the explicit definition of the sets $A_3$, $A_4$ and $A_5$. This implies that $A_1\cap  A_4$ is an empty set. Hence, $\hat{g}_0$, $\hat{g}_1$ and $\hat{g}_1/\hat{g}_0$ follow as defined by Theorem~\ref{theorem3}, \text{Type-III}. Following the same line of arguments for the cases $\hat{g}_0=\frac{1}{k_2} g_1^U$ on $ A_2$ and $\hat{g}_1=k_1 g_1^U$ on $A_5$ we have
\begin{align*}
A_1\cap A_5&=\{g_1^L<g_1<g_1^U,g_0=g_0^U\}=\{g_1^L/g_0^U<k_1=g_1/g_0<g_1^U/g_0^U\}, \nonumber\\
A_2\cap A_4&=\{g_0^L<g_0<g_0^U,g_1=g_1^U\}=\{g_1^U/g_0^U<k_2=g_1/g_0<g_1^U/g_0^L\},
\end{align*}
in the places of $A_0\cap  A_5$ and $A_2\cap  A_3$, respectively, empty $A_0\cap A_3$, and the explicit definition of the sets $A_j$, which leads to the LRF of \text{Type-I} and the corresponding LFDs. The LRF of \text{Type-II} is a special case arising from merging the middle three regions of the LRFs of \text{Type-I} and \text{Type-III} as $k_2\to k_1$. Moreover, LRFs of \text{Type-I} and \text{Type-III} tend to clipped likelihood ratio functions for $k_1$ small enough and $k_2$ large enough, and $k_1$ large enough and $k_2$ small enough, respectively. This implies empty $A_0\cap  A_4$ and $A_1\cap  A_3$.
\end{IEEEproof}

\begin{rem}
Three different types of LFDs given in Theorem~\ref{theorem3} were first proposed in \cite{kassamband} without any details about how they were obtained. Here, the robust LRFs and the corresponding LFDs have been derived analytically with the assumption that optimum Lagrangian parameters $\lambda_j$ and $\nu_j$ are constant functions. The correctness of these assumptions is due to Theorem~$1$ and $2$ in \cite{kassamband} which show that these pairs of LFDs are minimax robust and minimize all $f$-divergences.
\end{rem}
There are two different cases of consideration. If the type of LRFs are/can be known, it may be preferable to determine the parameters $k_1$ and $k_2$ by solving the equations which impose unit area density function constraints. However, if this knowledge is unavailable the minimax equations described by the problem formulation in Section~\ref{sec5_5} may need to be solved by a convex optimization method for the uncertainty classes given by \eqref{eq84}, which introduce linear constraints. To do this, the densities are first sampled, and hence discretized. For any integration, a numerical integration method can be adopted, for instance the trapezoidal integration.

\subsection{Moment Classes}
A common approach to partial statistical modeling is through moments, typically mean and correlation. The moment classes, which were originally introduced in \cite{moment}, can be generalized as
\begin{equation}\label{eq84a}
\mathscr{G}_j=\left\{G_j\in\mathscr{M}: c_{j,0}^k\leq \mathbb{E}_{G_j}[h_j^k(Y)] \leq c_{j,1}^k\right\},\quad k\in\{1,\ldots,K\},
\end{equation}
where $c_{j,0}^k$ and $c_{j,1}^k$ are some constants, $h_j^k$ are real valued continuous functions and $K$ is the total number of constraints. The constants and the functions should be chosen such that $\mathscr{G}_0\cap \mathscr{G}_1=\emptyset$. The original version of the moment classes have been studied for asymptotically minimax robust Neyman-Pearson tests in \cite{moment}. With the theory introduced so far it is now possible to obtain asymptotically minimax robust tests for the generalized version of the moment classes both in Neyman-Pearson as well as in rate minimizing sense.

\subsection{P-point Classes}
The partial information available for the robust hypothesis testing may also be in the form of probability masses which are assigned to every non-overlapping subsets of $\Omega$. The original definition of p-point classes can be extended covering a more general case as follows:
\begin{equation}\label{eq85a}
\mathscr{G}_j=\left\{G_j\in\mathscr{M}: c_{j,0}^k \leq G_j(A_j^k)\leq c_{j,1}^k\right\},\quad k\in\{1,\ldots,K\},
\end{equation}
where $A_j^k\in\mathscr{A}$ are some disjoint subsets of $\Omega$. The model by \eqref{eq85a} together with the theory introduced so far generalizes \cite{elsawy,elsawy2}.

\begin{rem}
For moment classes as well as p-point classes the LFDs are determined by first discretizing the domain of density functions and then solving the minimax equations described by the problem formulation in Section~\ref{sec5_5} for the uncertainty classes defined by \eqref{eq84a} or \eqref{eq85a}. It may also be of interest to combine p-point classes with moment classes in a hybrid model.
\end{rem}

So far seven different uncertainty classes have been introduced for which the design process either be performed by solving non-linear equations (all distance based uncertainty classes and the band model) or by solving a convex optimization problem (moment classes, p-point classes and the band model). Algorithm~\ref{algorithm1} summarizes the asymptotic minimax testing process with respect to these two cases, which are denoted by A and B.

\begin{algorithm}[ttt]\label{algorithm1}
\caption{Asymptotically minimax robust test design}
\vspace{1mm}
\SetAlgoLined
\KwIn{\hspace{3mm}A: $f_0,f_1,D,\epsilon_0,\epsilon_1,t$ \textbf{or} $g_{0}^L, g_{0}^U, g_{1}^L, g_{1}^U,t$ \\
\hspace{15mm} B: $\mathscr{G}_0$, $\mathscr{G}_1$ and $t$}
\KwOut{LFDs $\hat{g}_0$ and $\hat{g}_1$ of the asymptotically minimax robust test}
\Case{A}{\vspace{1mm}
Consider the LFDs $\hat{g}_0$ and $\hat{g}_1$ in parametric forms\\
Solve related equations and determine $\hat{g}_0$ and $\hat{g}_1$
}
\Case{B}{\vspace{1mm}
Discretize the domain of probability density functions $\Omega$\\
Determine $\hat{g}_0$ and $\hat{g}_1$ by solving \eqref{eq47} or \eqref{eq60} numerically
}
\Return $\frac{1}{n}\sum_{k=1}^{n}\log\hat{l}(y_k)\stackrel{\mathcal{H}_1}{\underset{\mathcal{H}_0}{\gtreqless}}t$
\end{algorithm}

\section{Simulations}\label{sec8}
In this section, the theoretical findings are evaluated and exemplified. For solving all systems of equations damped Newton's method \cite{ralph} and for all convex optimization problems interior point methods \cite{potra} are used. To make the simulations transparent and easily repeatable the parameter values are explicitly stated. The notation $|_a^b$ stands for testing with the $(a)$-test while the data samples are obtained from the LFDs of the $(b)$-test. In all theoretical examples, the nominal distributions listed in Table~\ref{tab1} are considered. The notation $\mathcal{N}(\mu,\sigma^2)$ stands for the Gaussian distribution with mean $\mu$ and variance $\sigma^2$ whereas $\mathcal{L}(0,1)$ denotes the standard Laplace distribution with the respective parameters. The density functions are similarly denoted by $f_{\mathcal{N}}$ and $f_\mathcal{L}$, respectively. In the following, the least favorable distributions, robust likelihood ratio functions, parameters of the equations, and (non)-convexity of $D_u$ are illustrated.
\begin{table}[ttt]
\caption{Pair of nominal distributions used in the simulations}
\begin{center}
\begin{tabular}{|c|l|l|}
\hline
Acronym & Under $\mathcal{H}_0$ & Under $\mathcal{H}_1$ \\
\hline \hline
$d_1$ & $\mathcal{N}(-1,1)$ & $\mathcal{N}(1,1)$  \\
\hline
$d_2$ & $\mathcal{N}(-1,1)$ & $\mathcal{N}(1,4)$ \\
\hline
$d_3$ &  $\mathcal{L}(0,1)$ & $f_{\mathcal{L}}(y)(\sin(2\pi y)+1)$   \\
\hline
\end{tabular}
\end{center}
\label{tab1}
\end{table}
\begin{figure}[ttt]
  \centering
  \centerline{\includegraphics[width=8.8cm]{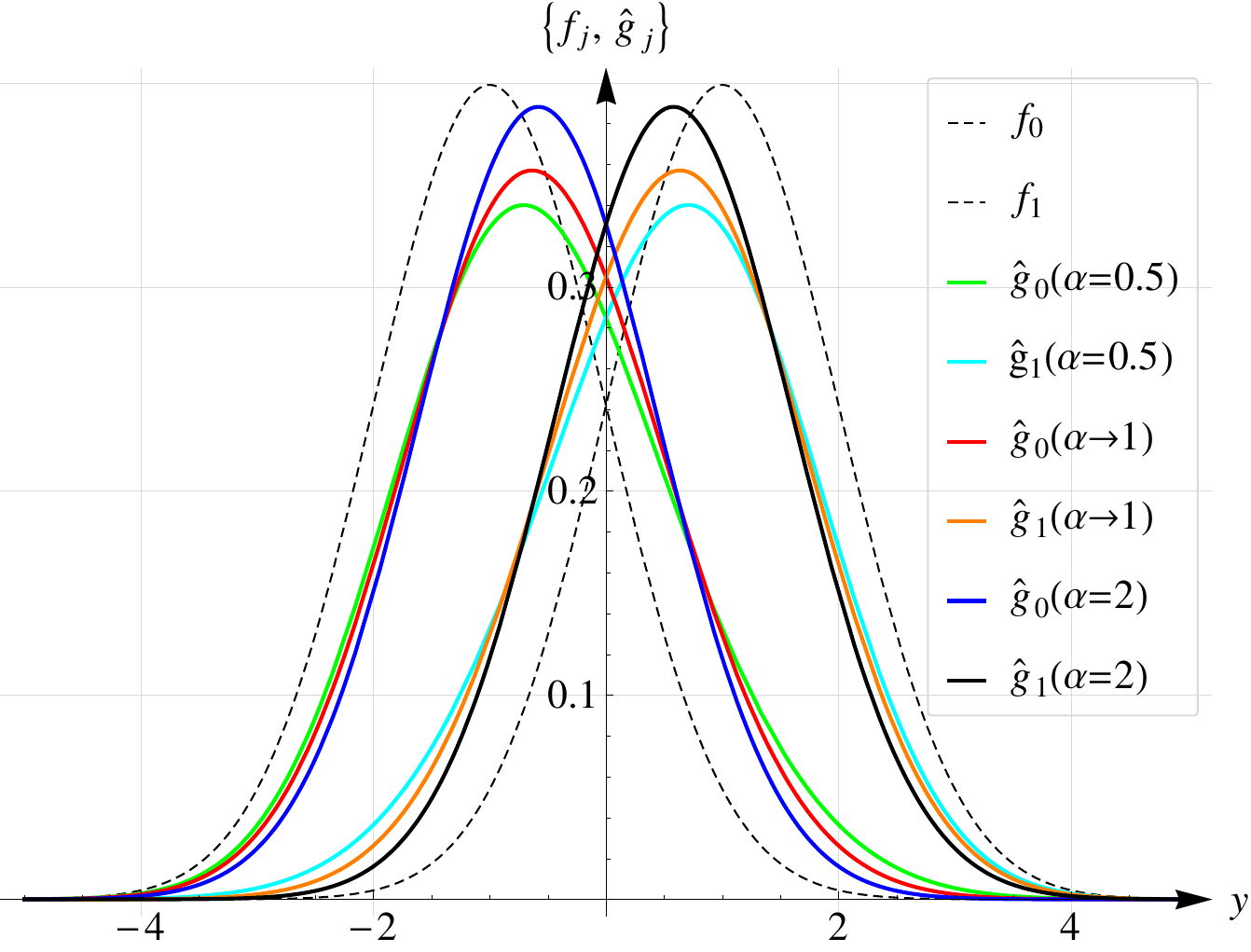}}
\caption{The nominal distributions denoted by $d_1$ and the corresponding LFDs for $\epsilon_0=\epsilon_1=0.1$, where $u=0.5$ for all $\alpha$.\label{fig1}}
\end{figure}
\begin{figure}[ttt]
  \centering
  \centerline{\includegraphics[width=8.8cm]{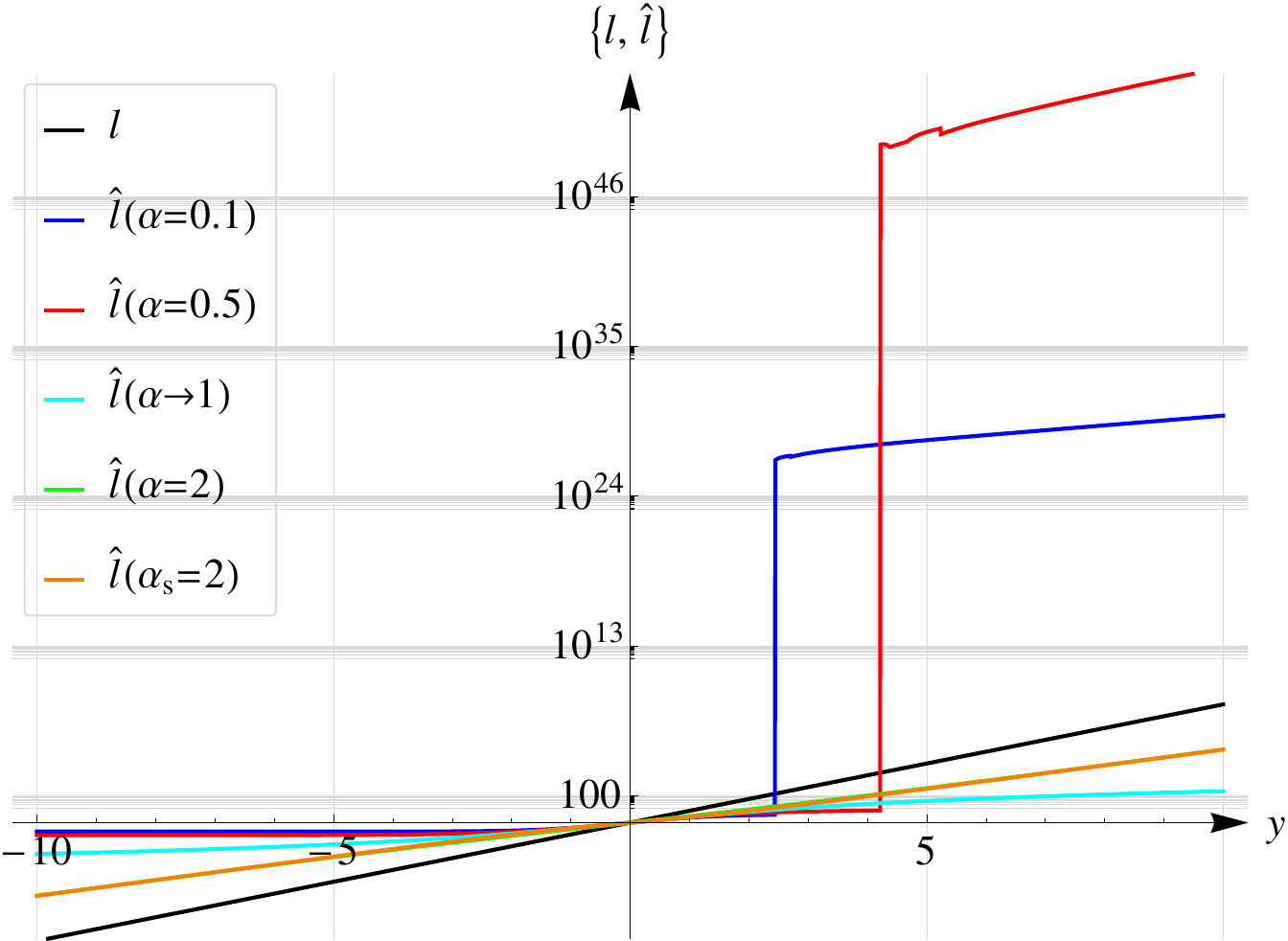}}
\caption{Robust and nominal LRFs found for the nominal distributions denoted by $d_1$ and $\epsilon_0=\epsilon_1=0.1$, including the symmetric case $\alpha_s$.\label{fig2}}
\end{figure}
\subsection{LFDs and Robust LRFs}
Comparative simulations are required in order to get the intuition about how robustness is achieved depending on the choice of the distance. Consider the pair of distributions denoted by $d_1$ in Table~\ref{tab1} and let the robustness parameters be $\epsilon_0=\epsilon_1=0.1$. For this setup, Figure~\ref{fig1} illustrates the LFDs together with the nominal distributions for the KL-divergence ($\alpha\rightarrow 1$) as well as for various $\alpha$-divergences. Symmetrized $\alpha$-divergence is not included for the sake of clarity. There are two observations from this example:
\begin{itemize}
\item The LFDs are non-Gaussian (not visible but verified by means of curve fitting).
\item The variance of the LFDs are decreasing as $\alpha$ increases.
\end{itemize}
In Figure~\ref{fig2} the corresponding likelihood ratio functions are depicted, including the symmetrized $\alpha=2$-divergence, denoted by $\alpha_s=2$. For $\alpha=0.1$ and $\alpha=0.5$ there is a strong amplification of the likelihood ratios for larger observations and clipping for smaller observations (not well visible) in order to achieve asymptotic robustness. Moreover, there is no recognizable difference between the LRFs of $\alpha=2$ and $\alpha_s=2$.\\
The pair of nominal distributions $d_1$ are symmetric about the origin and asymmetric nominals are known to complicate the solution of the non-linear equations \cite{gul6}. Additionally, the LRFs can be visualized in a reduced range, focusing more on the clipping range rather than the range of amplification, to be complimentary to the previous example. In this regard, the pair of nominal distributions denoted by $d_2$ in Table~\ref{tab1} are considered. Figure~\ref{fig3} illustrates the LFDs together with the nominals whereas Figure~\ref{fig4} shows the corresponding robust LRFs for $\epsilon_0=\epsilon_1=0.1$. For larger observations in absolute value, there is huge amplification (not well visible), whereas for the smaller observations there is no hard clipping as in the previous example. The difference between $\alpha=2$ and the symmetrized $\alpha_s=2$ divergences is now visible.
\begin{figure}[ttt]
  \centering
  \centerline{\includegraphics[width=8.8cm]{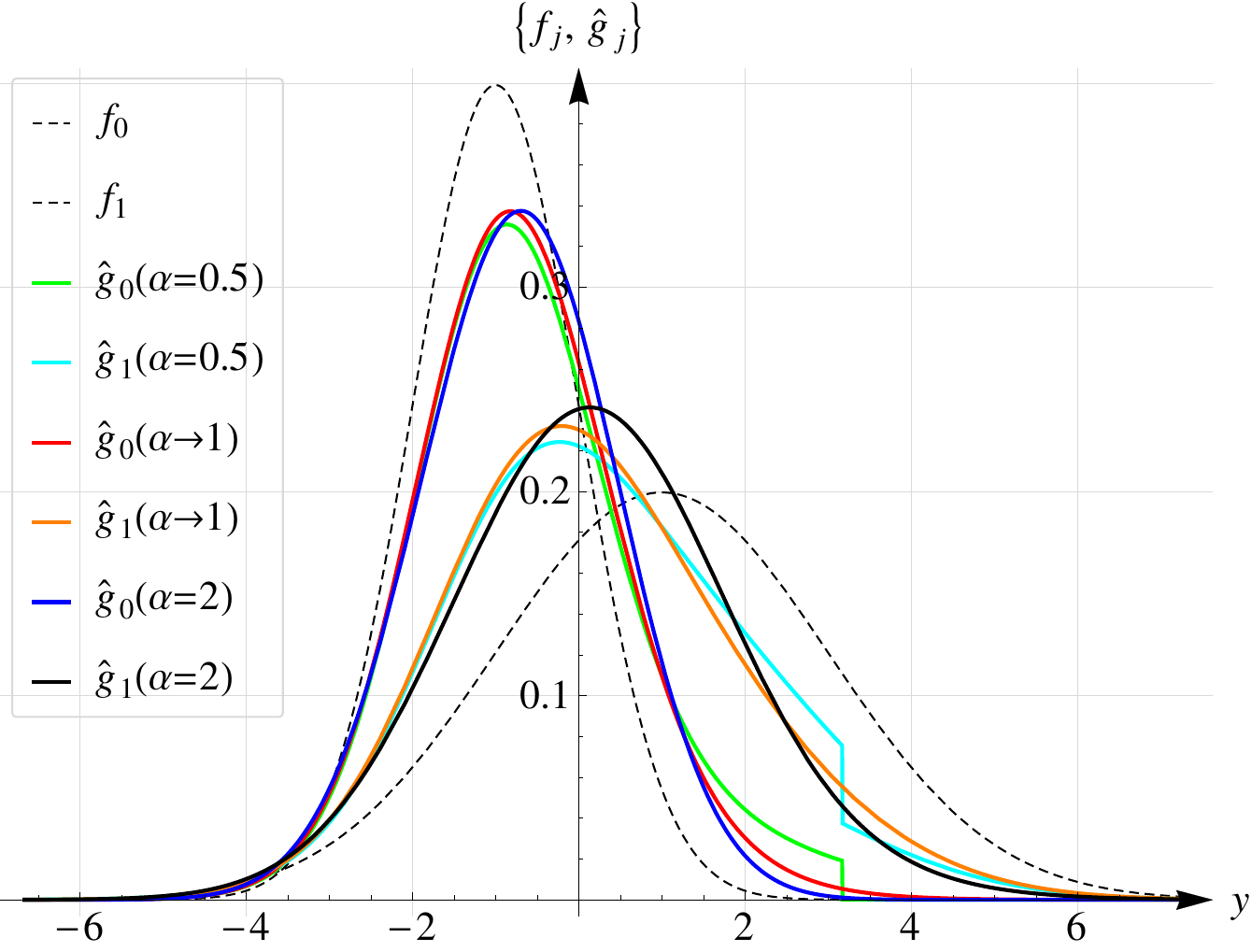}}
\caption{The nominal distributions denoted by $d_2$ and the corresponding LFDs for $\epsilon_0=\epsilon_1=0.1$.\label{fig3}}
\end{figure}
\begin{figure}[ttt]
  \centering
  \centerline{\includegraphics[width=8.8cm]{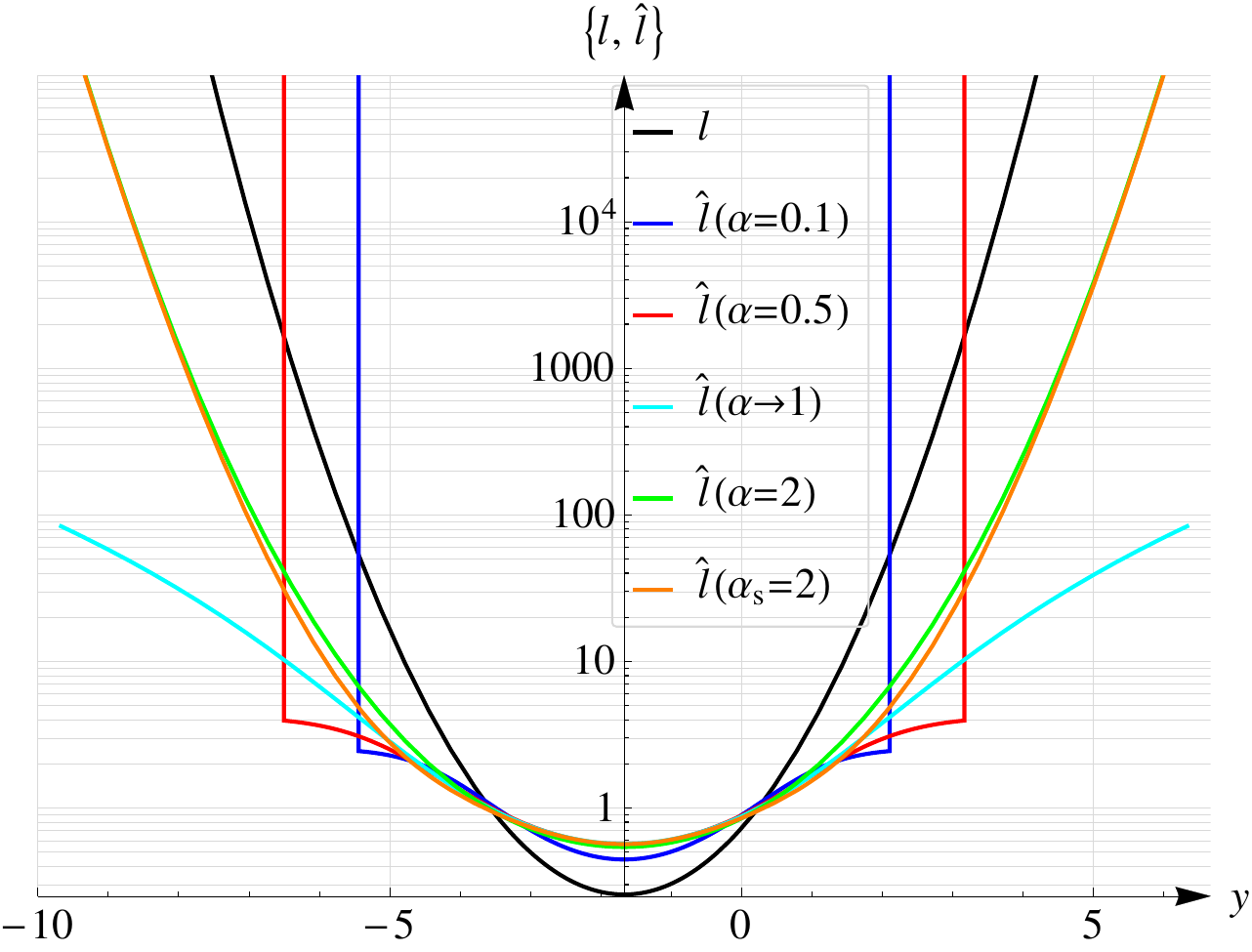}}
\caption{Robust and nominal LRFs found for the nominal distributions denoted by $d_2$ and $\epsilon_0=\epsilon_1=0.1$, including the symmetric case $\alpha_s$. The optimum values of $u$ are $0.95,0.67,0.56,0.59,0.61$, respectively, from $\alpha=0.1$ to $\alpha_s=2$.\label{fig4}}
\end{figure}\\
The nominal LRFs are either increasing, or first decreasing and then increasing, respectively, for the pair of distributions denoted by $d_1$ and $d_2$. It is possible to construct an example for which the nominal LRF is repeatedly increasing and decreasing. This case both confirms the solvability of the related non-linear equations and serves as an example for the convexity analysis in the next section. Let the nominal distributions be denoted by $d_3$ as given in Table~\ref{tab1}. Furthermore, let $\epsilon_0=\epsilon_1=0.05$, as the nominal distributions are now closer to each other. For this setup, Figure~\ref{fig5} and Figure~\ref{fig6} illustrate the LFDs together with the nominals and the robust LRFs, respectively, for the KL-divergence neighborhood. Similar to the previous examples, the nominal LRFs which are smaller than $1$ are amplified and those larger than $1$ are attenuated.

\begin{figure}[ttt]
  \centering
  \centerline{\includegraphics[width=8.8cm]{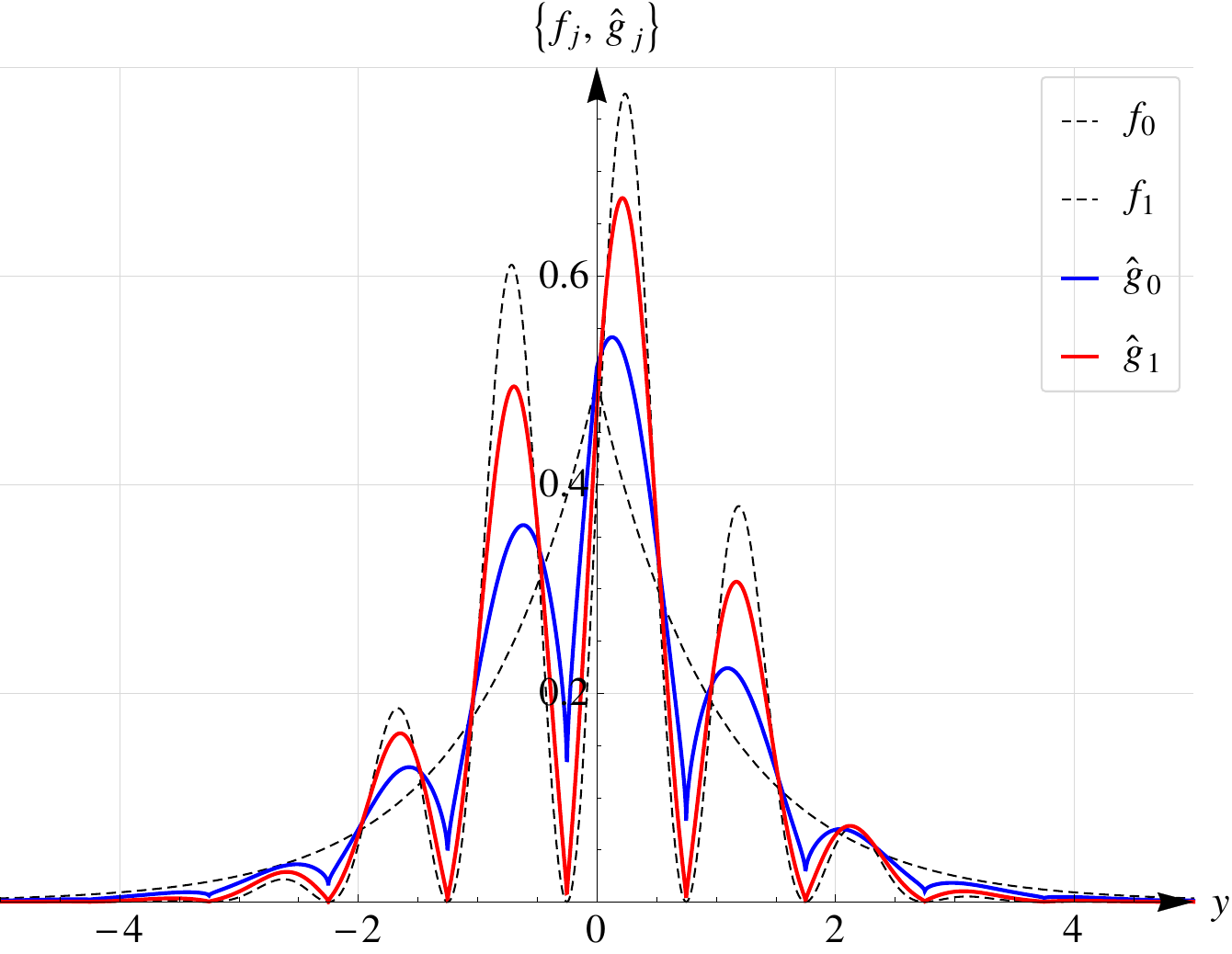}}
\caption{The nominal distributions denoted by $d_3$ and the corresponding LFDs for $\epsilon_0=\epsilon_1=0.05$, where $u=0.46$.\label{fig5}}
\end{figure}

\begin{figure}[ttt]
  \centering
  \centerline{\includegraphics[width=8.8cm]{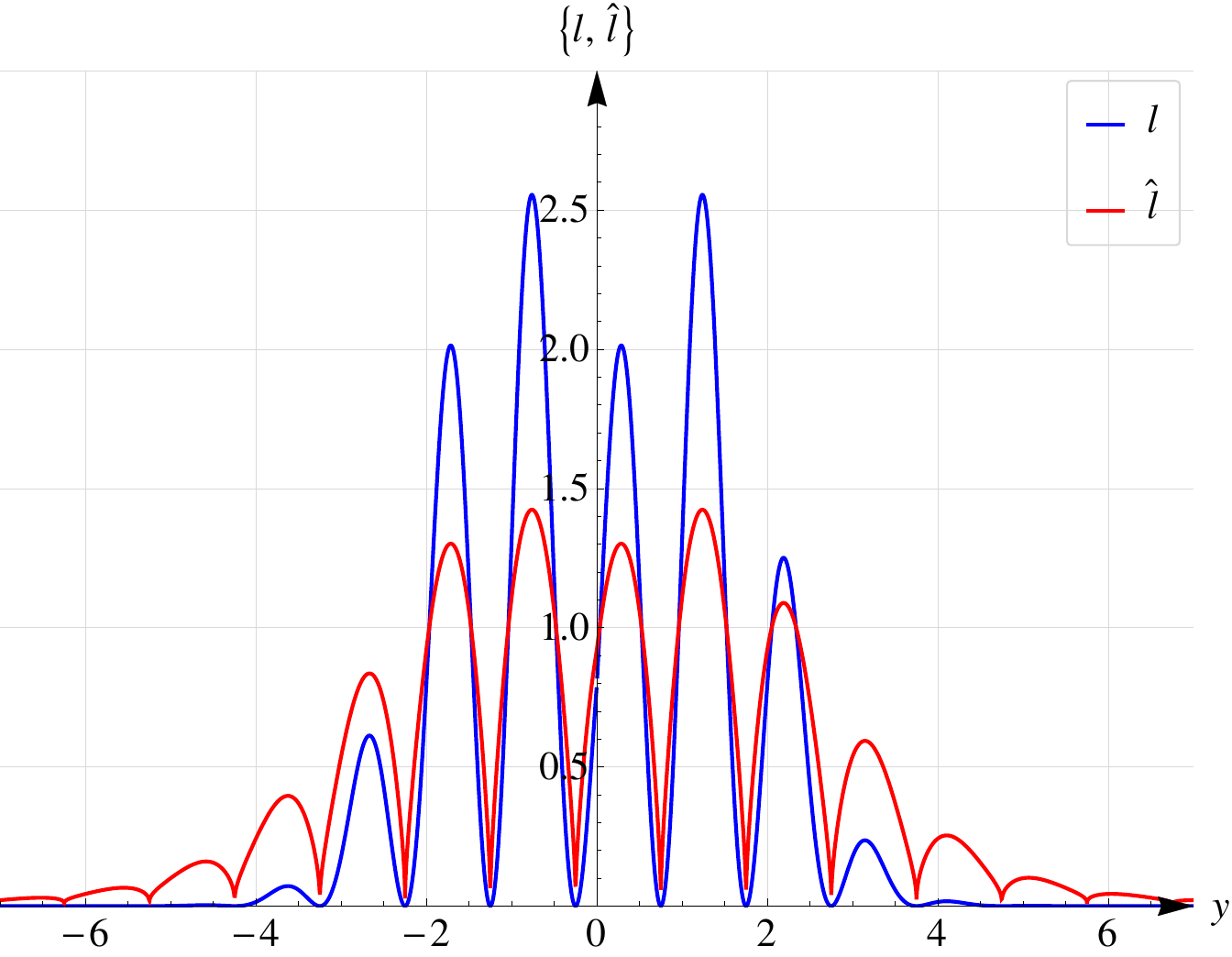}}
\caption{Robust and nominal LRFs found for the nominal distributions denoted by $d_3$ and $\epsilon_0=\epsilon_1=0.05$.\label{fig6}}
\end{figure}

\subsection{Convexity of $D_u$ and the Lagrangian Parameters}
It was mentioned earlier that $D_u$ is convex in $u$, for a fixed pair of distributions. However, it is not necessarily convex if for every $u$ the distribution functions are possibly different. This is especially the case when one considers the LFDs which are found as a function of $u$, cf. Section~\ref{sec5}. In order to see whether the convexity arguments still hold in general, $D_u$ is plotted for the pair of nominal distributions denoted by $d_1$, $d_2$ and $d_3$, when the distance is the \text{KL-divergence} and additionally for $d_2$, when the distance is the $\alpha=0.1$-divergence. The robustness parameters are the same as in the previous simulations. Figure~\ref{fig7} illustrates the outcome of this simulation, which proves the existence of distances (i.e. $\alpha=0.1$) for which $D_u$ is not necessarily convex, although it may not possibly be the case for the KL-divergence.\\
The LFDs are obtained by solving a system of non-linear equations for every choice of $u$. These parameters can be depicted so that the results can easily be verified by others. In this example again the KL-divergence neighborhood is considered with $\epsilon_0=\epsilon_1=0.1$. In Figure~\ref{fig8} the KKT parameters $\lambda_0$, $\lambda_1$, $\mu_0$ and $\mu_1$ are illustrated for the pairs of nominal distributions denoted by $d_1$ and $d_2$. For both examples, the KKT parameters follow similar paths.

\begin{figure}[ttt]
  \centering
  \centerline{\includegraphics[width=8.8cm]{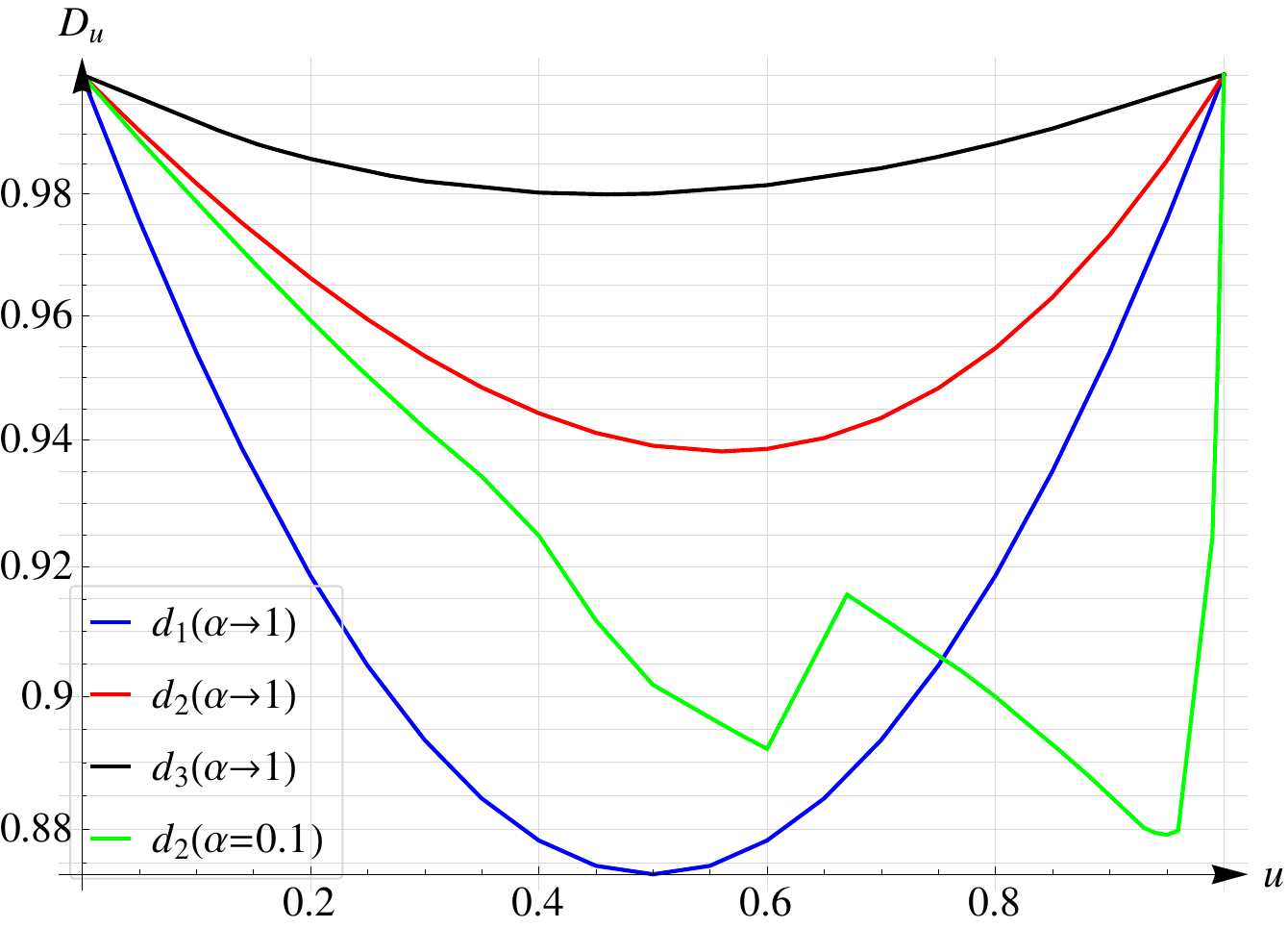}}
\caption{The $u$-divergence as a function of $u$ for the LFDs obtained for various pairs of distributions as well as uncertainty classes.\label{fig7}}
\end{figure}

\begin{figure}[ttt]
  \centering
  \centerline{\includegraphics[width=8.8cm]{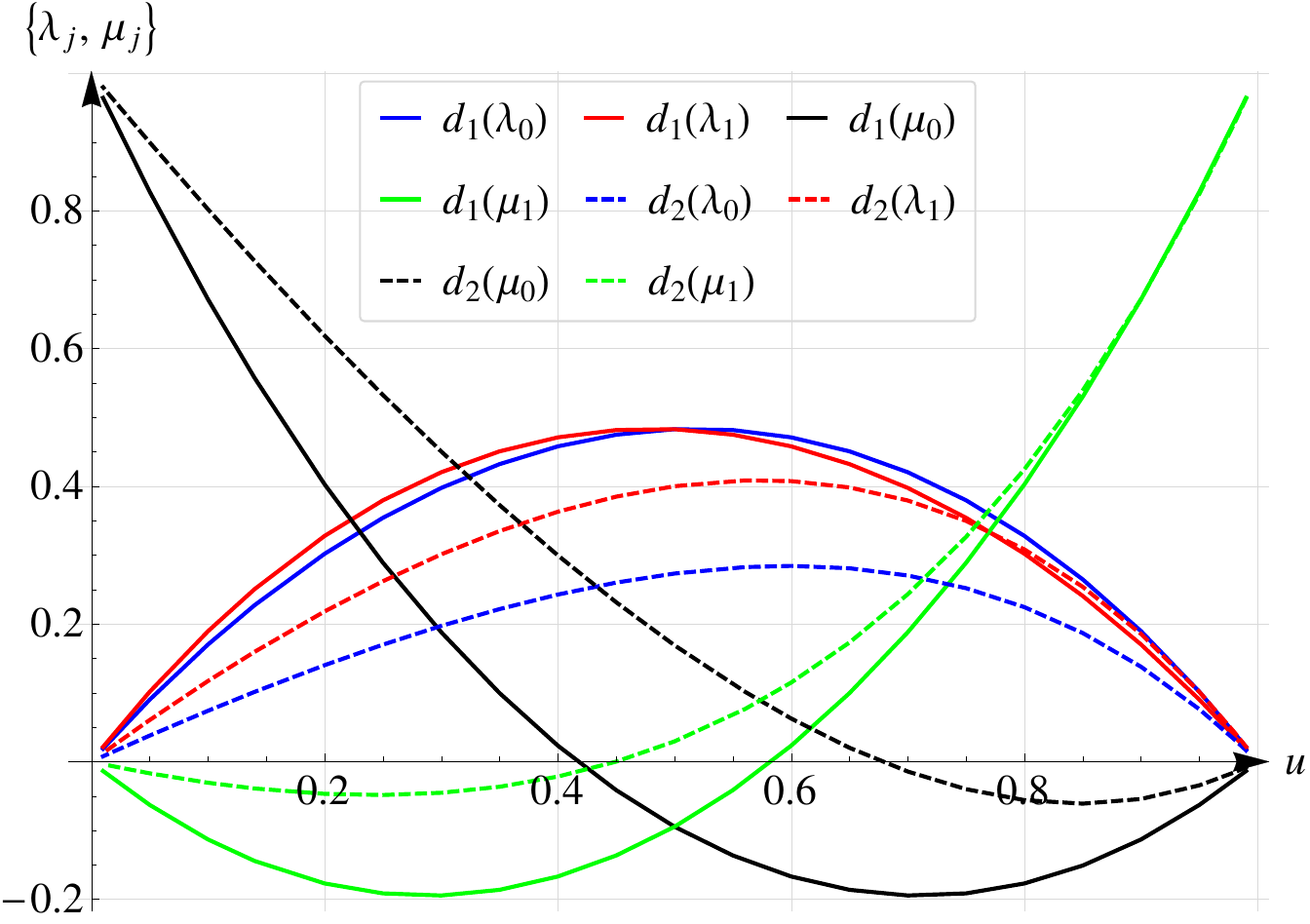}}
\caption{KKT parameters for two pairs of nominal distributions and the KL-divergence neighborhood with $\epsilon_0=\epsilon_1=0.1$.\label{fig8}}
\end{figure}

\subsection{Asymptotically Minimax Robust NP-tests}
Dabak's test is neither minimax robust nor asymptotically minimax robust as it was shown in Section~\ref{sec5},~\ref{sec6_1}. This result can be demonstrated with an example. Let the nominal distributions be given as in Table~\ref{tab1} with $d_2$ and let $\epsilon_0=\epsilon_1=0.01$. The rate functions $I_0$ and $I_1$ are of interest for two cases; $|_a^\cdot$ and $|_{a^*}^\cdot$, i.e. when the test is asymptotically minimax robust Type-I NP-test and Dabak's test, respectively. Figures~\ref{fig15} and \ref{fig16} illustrate the rate functions $I_0$ and $I_1$. Both in Figure~\ref{fig15} and Figure~\ref{fig16} the performance of the ($a^*$)-test is degraded by the data samples obtained from the LFDs of the $(a)$-test in comparison to those obtained from the LFDs of the ($a^*$)-test.

\subsection{Band Model}
Asymptotically minimax robust tests arising from the band model can similarly be simulated. Consider the lower bounding functions
\begin{equation*}
g_0^L(y)=(1-\epsilon)f_{\mathcal{N}}(y;-1,4),\quad g_1^L(y)=(1-\epsilon)f_{\mathcal{N}}(y;1,4),
\end{equation*}
where the contamination ratio is chosen to be $\epsilon=0.2$. Furthermore, let the upper bounding functions be
\begin{equation*}
g_0^U(y)=(1+\varepsilon)f_{\mathcal{N}}(y;-1,4),\quad g_1^U(y)=(1+\varepsilon)f_{\mathcal{N}}(y;1,4),
\end{equation*}
with the parameters $\varepsilon=0.2$ (Type-I), $\varepsilon=0.5$ (Type-II), $\varepsilon=1.5$ (Type-III) or $\varepsilon=19$ (Type-III), simulating three different types of robust LRFs resulting from the band model, cf. Section~\ref{sec6_band}.\\
For this setup, and excluding $\varepsilon=19$ for the sake of clarity, Figure~\ref{fig9} illustrates the corresponding LFDs together with the lower bounding functions, and the upper bounding functions for $\varepsilon=1.5$. For $\varepsilon=0.5$, the LFDs are overlapping around $y=0$, leading to $\hat{l}=1$. This type of overlapping has previously been reported by \cite{gul6} for single-sample minimax robust tests obtained from the KL-divergence neighborhood. However, the test in \cite{gul6} is not minimax robust unless a well defined randomized decision rule is used.\\
In Figure~\ref{fig10}, the corresponding robust likelihood ratio functions are illustrated. Increasing $\varepsilon$ transforms the corresponding robust LRF from Type-I to Type-II and then to Type-III. Further increasing $\varepsilon$, i.e. when $\varepsilon=19$, the robust LRF tends to a clipped likelihood ratio test, which is the limiting LRF stated in Section~\ref{sec6_band}. The robust LRFs can take different shapes depending on the bounding functions. Similar patterns were stated in \cite{kassamband} and also observed in \cite{fauss}.

\begin{figure}[ttt]
  \centering
  \centerline{\includegraphics[width=8.8cm]{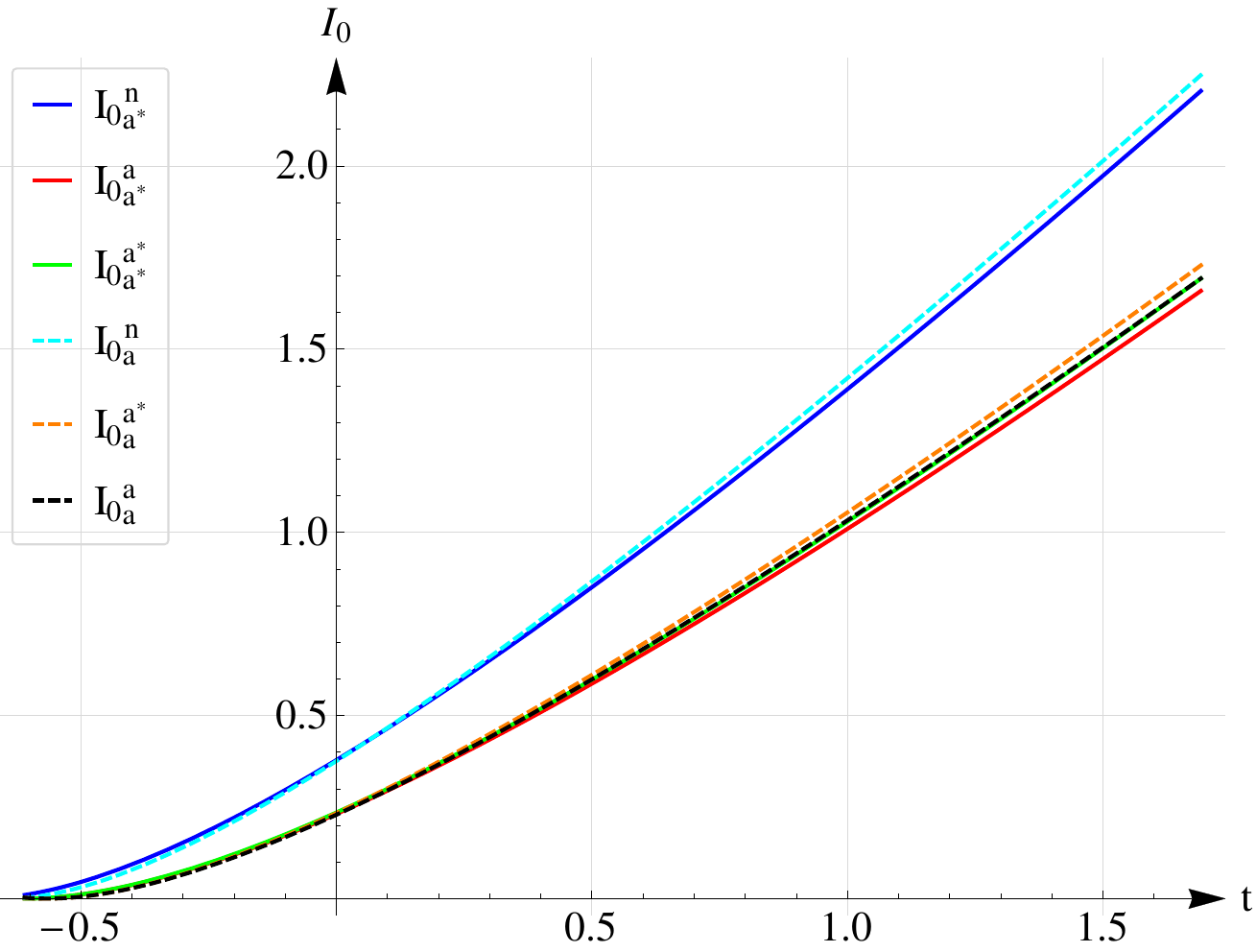}}
\caption{Asymptotic decay rate of the false alarm probability $I_0$ for the asymptotically minimax robust NP-test ((a)-test) and Dabak's test ((a$^*$)-test).\label{fig15}}
\end{figure}

\begin{figure}[ttt]
  \centering
  \centerline{\includegraphics[width=8.8cm]{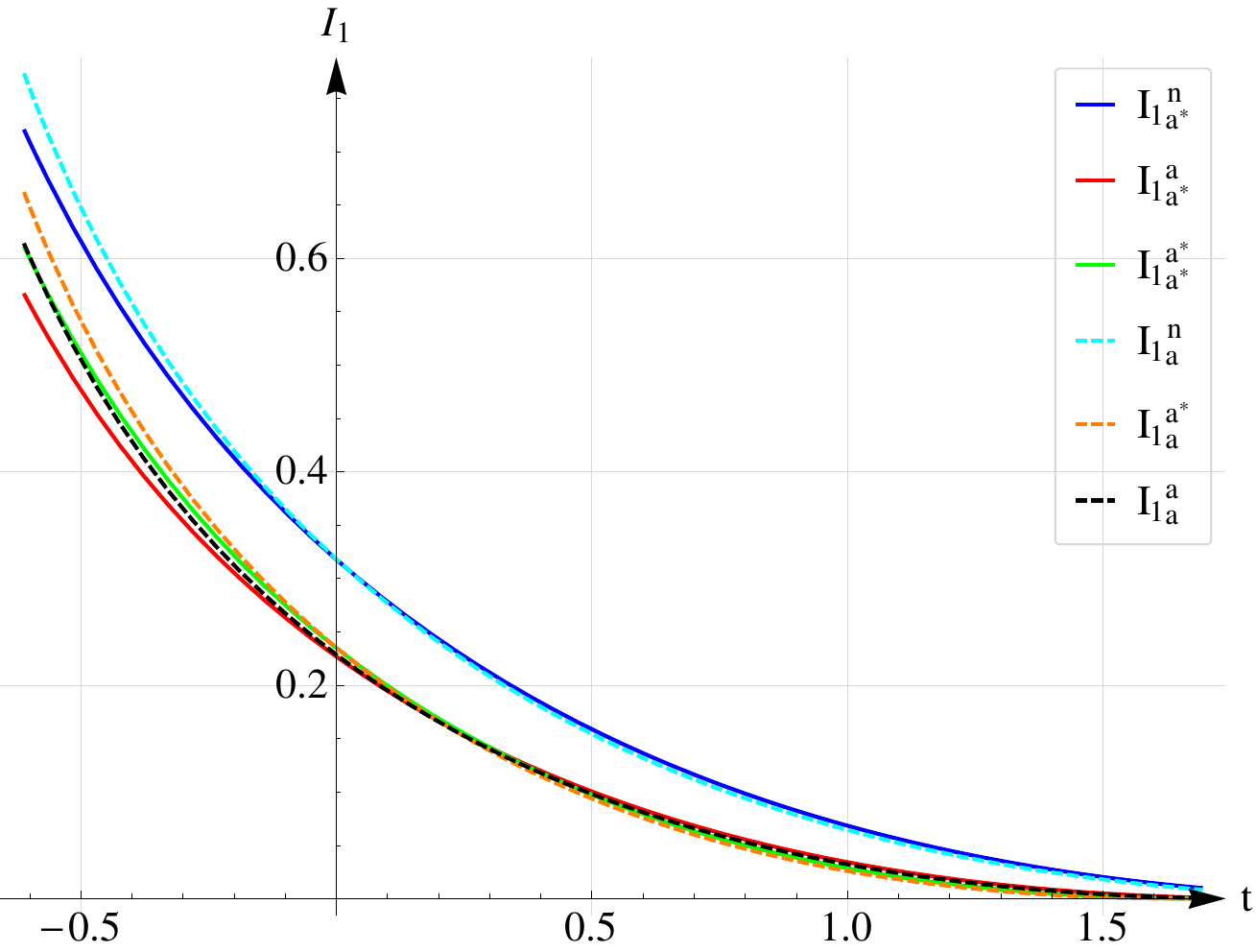}}
\caption{Asymptotic decay rate of the miss detection probability $I_1$ for the asymptotically minimax robust NP-test ((a)-test) and Dabak's test ((a$^*$)-test).\label{fig16}}
\end{figure}

\begin{figure}[ttt]
  \centering
  \centerline{\includegraphics[width=8.8cm]{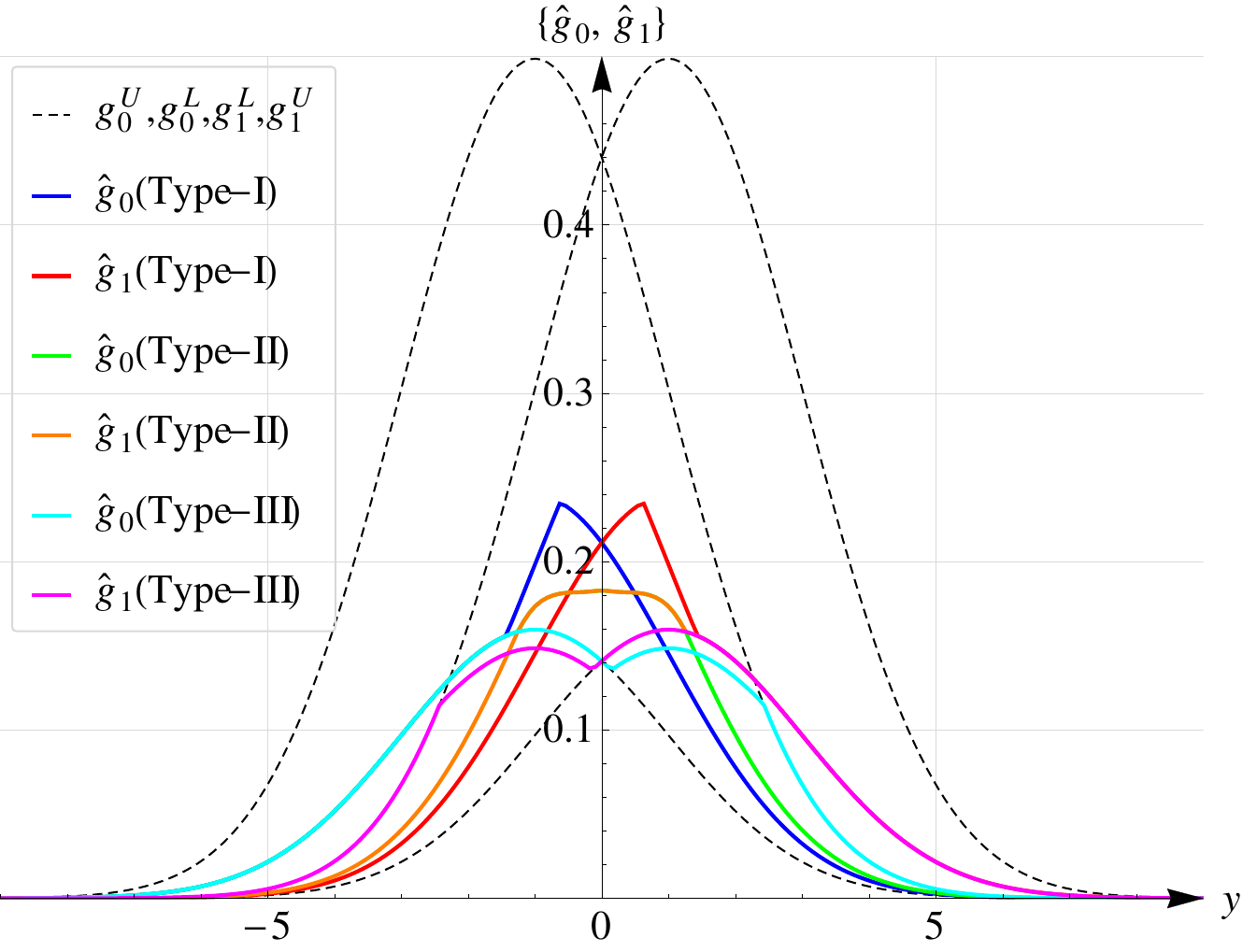}}
\caption{Three different pairs of LFDs arising from the band model together with the bounding functions for $\varepsilon\in\{0.2,0.5,1.5\}$.\label{fig9}}
\end{figure}

\begin{figure}[ttt]
  \centering
  \centerline{\includegraphics[width=8.8cm]{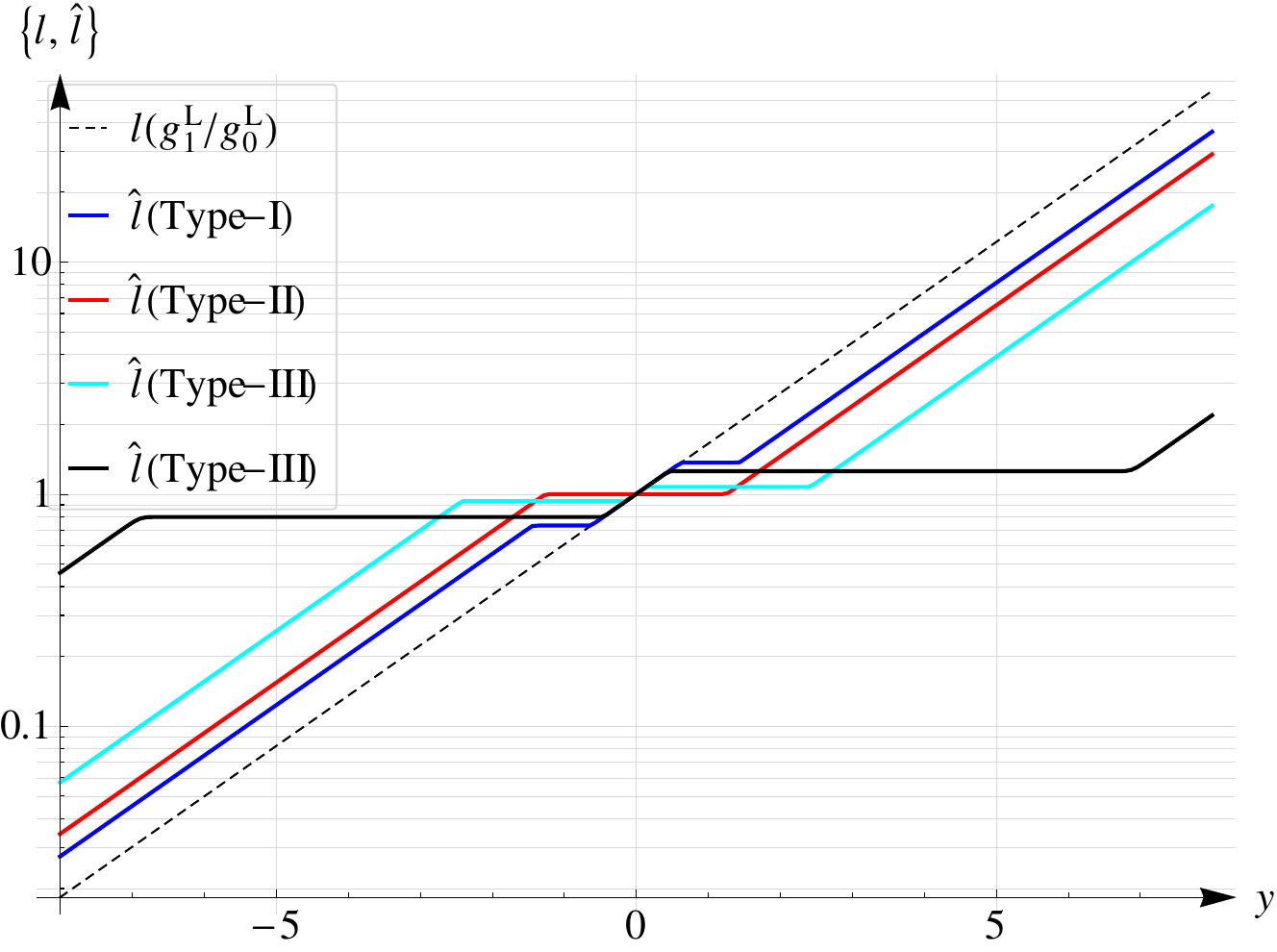}}
\caption{Three different types of Robust LRFs arising from the band model for $\varepsilon\in\{0.2,0.5,1.5,19\}$ together with the nominal LRF.\label{fig10}}
\end{figure}

\subsection{Moment Classes}
The LFDs and robust LRFs arising from the moment classes can be exemplified by solving the convex optimization problem given in Section~\ref{sec5_5} numerically, by first replacing the constraints with the ones defined by the moment classes. Consider the constraints
\begin{align*}
-2 \leq &\mathbb{E}_{G_0}[Y] \leq -0.5,& 0.5 \leq \mathbb{E}_{G_1}[Y] \leq 2,\nonumber \\
0 \leq &\mathbb{E}_{G_0}[Y^2] \leq 2,& 2 \leq \mathbb{E}_{G_1}[Y^2] \leq 4,
\end{align*}
\begin{figure}[ttt]
  \centering
  \centerline{\includegraphics[width=8.8cm]{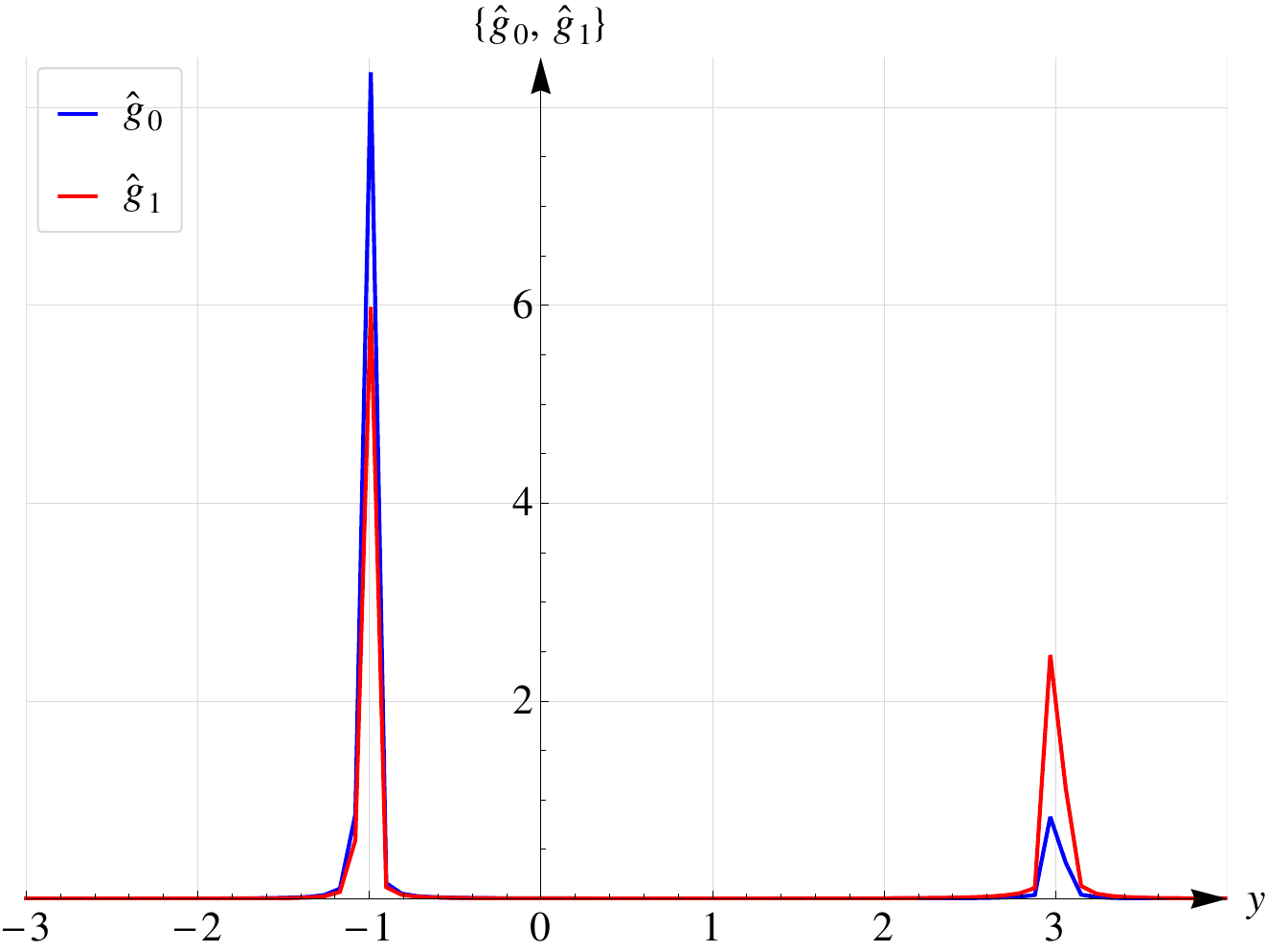}}
\caption{Least favorable distributions arising from the moment classes in the given example.\label{fig19}}
\end{figure}
\begin{figure}[ttt]
  \centering
  \centerline{\includegraphics[width=8.8cm]{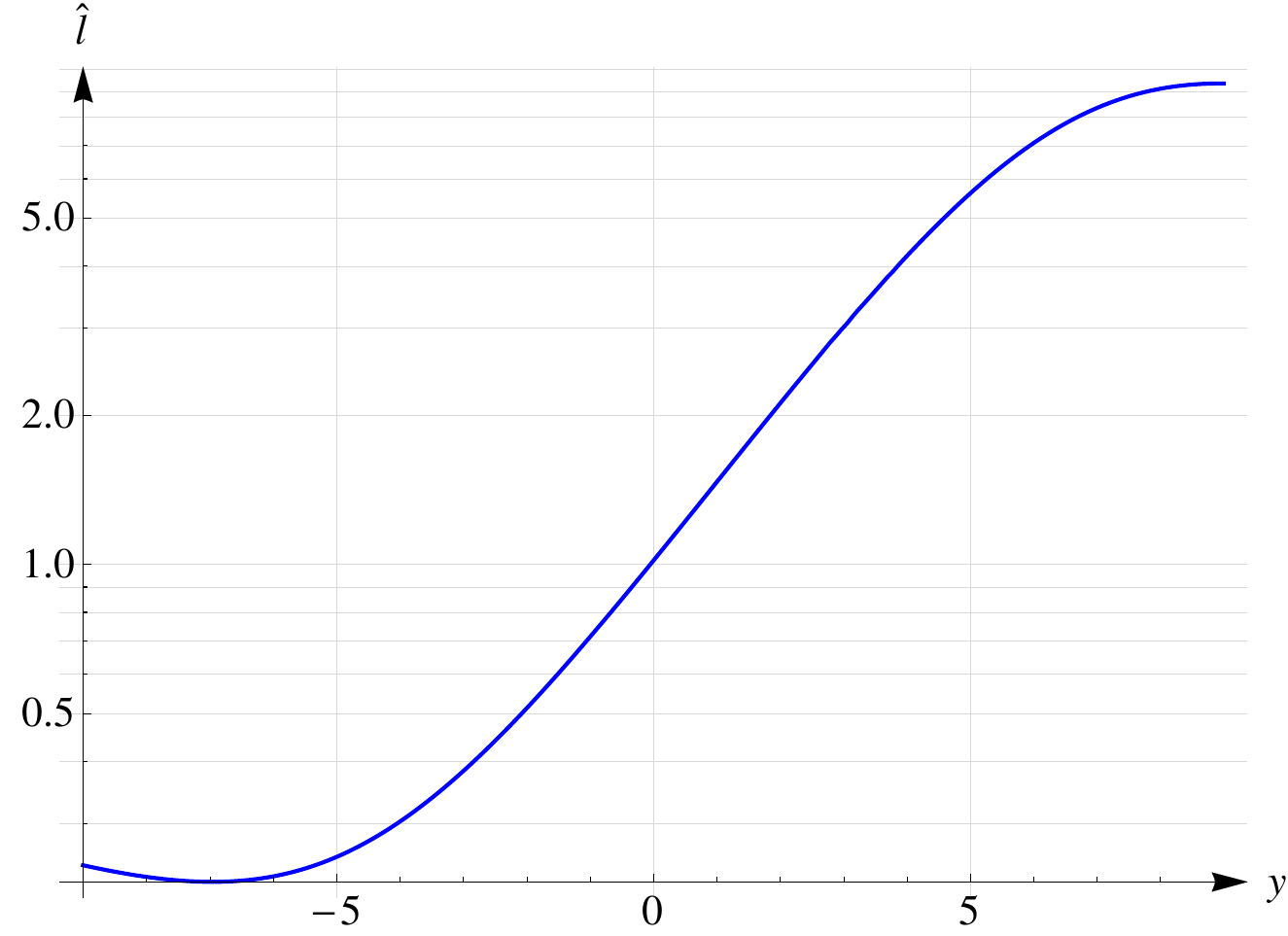}}
\caption{The robust LRF arising from the moment classes in the given example.\label{fig20}}
\end{figure}defined over the first and second moments of the probability density functions. Figures~\ref{fig19} and \ref{fig20} illustrate the LFDs and the corresponding robust LRF, respectively.

\subsection{P-point Classes}
Similarly, an example to the asymptotically minimax robust test arising from the p-point classes can be given. Consider the p-point classes defined by the constraints
\begin{align*}
\int_{-5}^{3} g_0(y)d y \leq 0.3,\quad \int_{0}^{3} g_1(y)d y \geq 0.8.
\end{align*}
Figures~\ref{fig21} and \ref{fig22} illustrate the LFDs and the corresponding robust LRF, respectively.

\begin{figure}[ttt]
  \centering
  \centerline{\includegraphics[width=8.8cm]{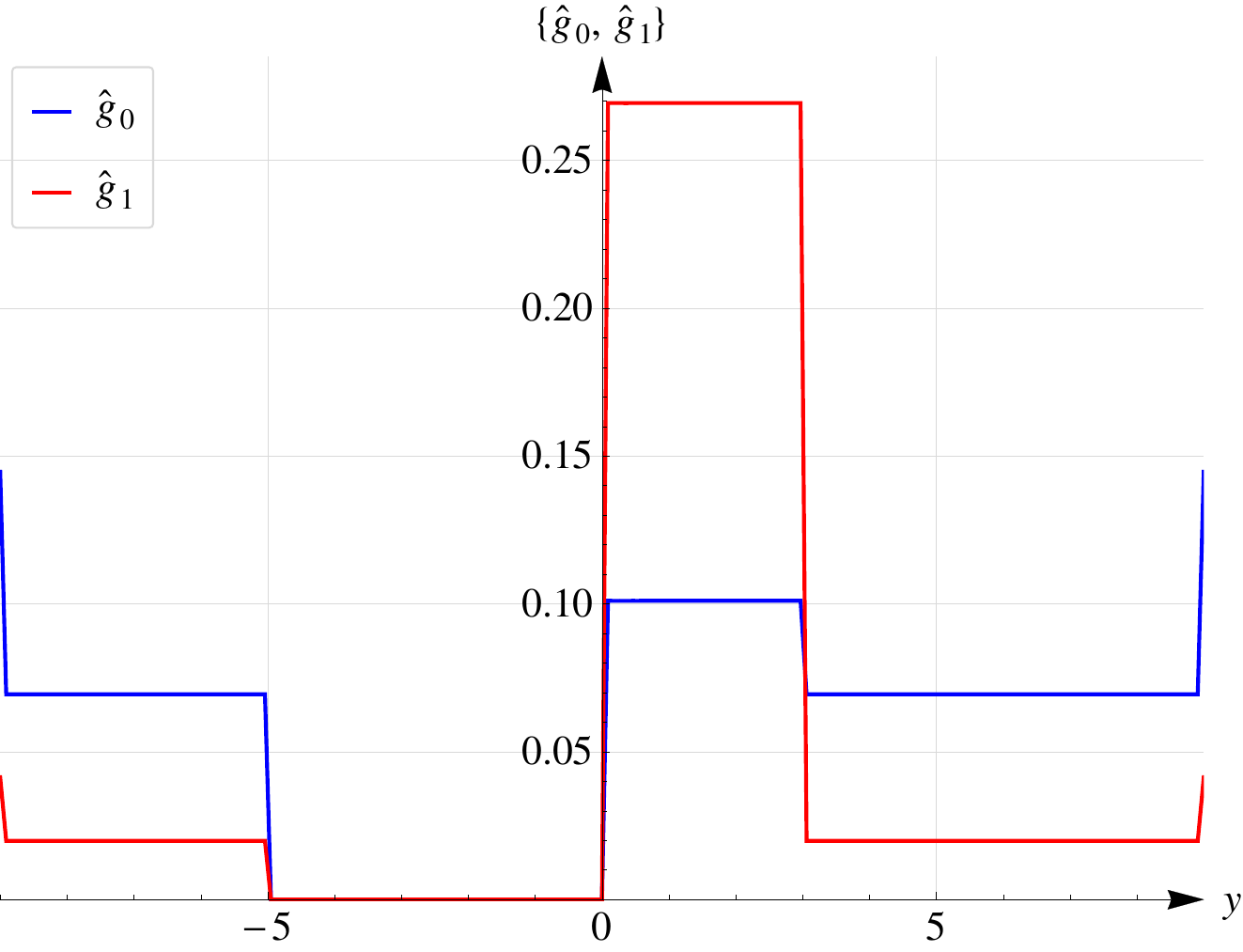}}
\vspace{-3mm}
\caption{Least favorable distributions arising from the p-point classes in the given example.\label{fig21}}
\end{figure}

\begin{figure}[ttt]
  \centering
  \centerline{\includegraphics[width=8.8cm]{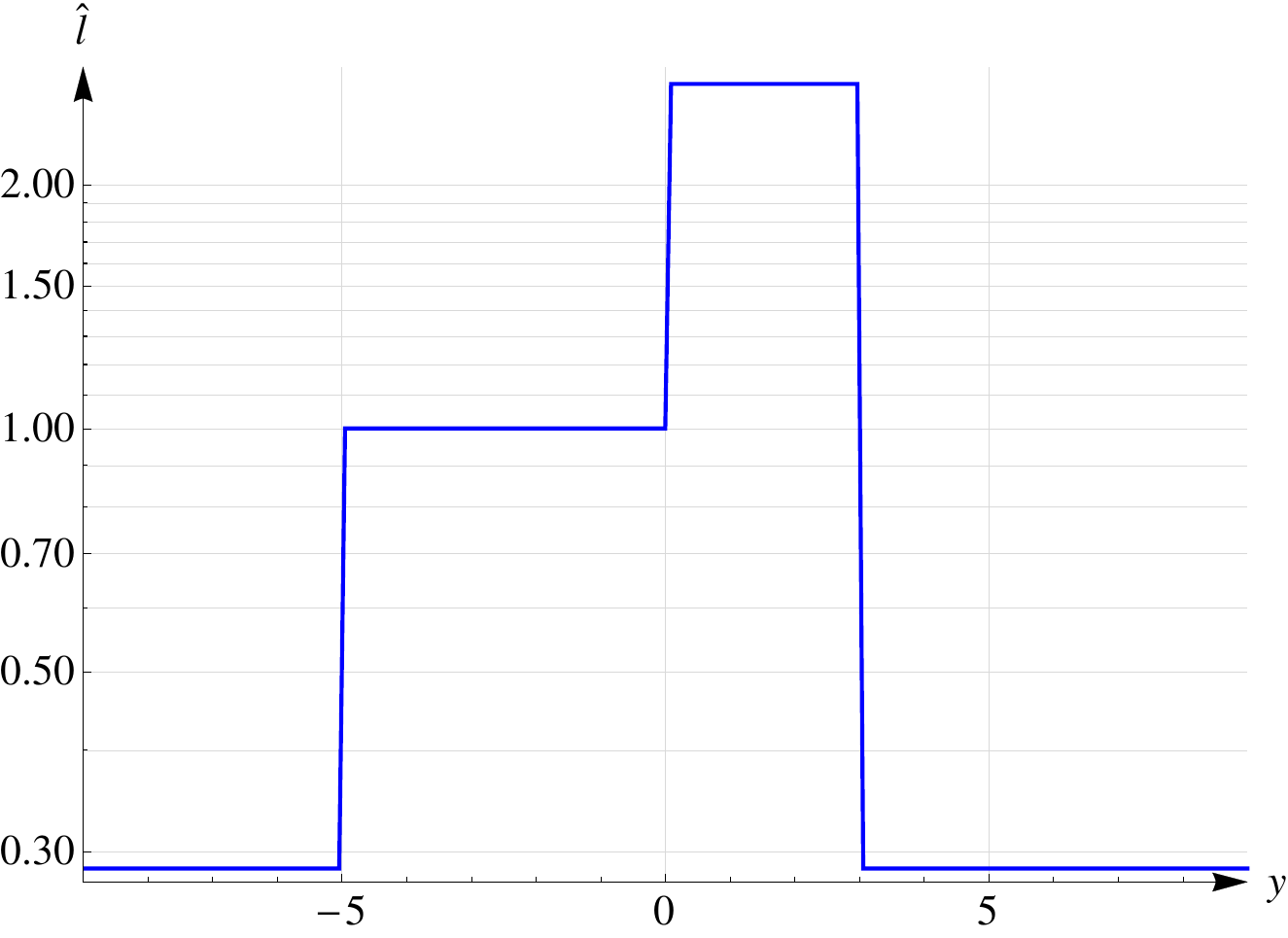}}
\vspace{-3mm}
\caption{The robust LRF arising from the p-point classes in the given example.\label{fig22}}
\end{figure}

\section{Conclusions}\label{sec9}
Designing (asymptotically) minimax robust hypothesis tests for any reasonable construction of uncertainty classes was shown to be made via maximizing $D_u(G_0,G_1)$ over all $(g_0,g_1)\in\mathscr{G}_0\times \mathscr{G}_1$ and minimizing $u$ over all $u\in(0,1)$. The uncertainty classes based on the KL-divergence, $\alpha$-divergence, symmetrized $\alpha$-divergence, total variation distance, as well as the band model, moment classes and p-point classes were considered. For the first six classes, the KKT multipliers were employed to derive the LFDs and the corresponding minimax robust tests in parametric forms. The related parameters can be determined by means of solving non-linear systems of equations. For the last two classes, asymptotically minimax robust tests were evaluated as a convex optimization problem.\\
In addition to the Bayesian formulation, Neyman-Pearson formulations of the asymptotically minimax robust hypothesis testing schemes were considered for the KL-divergence neighborhood. The resulting Neyman-Pearson tests of Type~I and II correspond to a non-linear transformation of the nominal LRF, which involves the Lambert-W function. This result proves that Dabak's test is not asymptotically minimax robust.\\
Existence of a deterministic single-sample minimax robust test implies existence of all-sample minimax robust tests. It is also known that a single-sample minimax robust test may fail to exist even for a simple example. Besides, randomized versions of single-sample minimax robust tests, if they exist, do not lead to finite-sample minimax robust tests. These results inevitably motivate considering asymptotical designs. Although the term \emph{asymptotic} implies validity of the theory for very large sample sizes, it was shown that the proposed theory allows finding the finite-sample minimax robust test, if it exists and  the asymptotically minimax robust test, otherwise. Therefore, in some sense, the asymptotic design may be called the best possible approach in terms of minimax robustness.\\
In order to evaluate theoretical findings and their applicability simulations were performed. In particular, the LFDs and the robust LRFs were exemplified for every uncertainty class considered. While clipping the nominal likelihood ratios has already been known to be a way of employing minimax robustness, it was first observed here that asymptotically minimax robustness may require amplifying the nominal likelihood ratios by a great factor. Additionally, a large number of constraints may be required such that the LFDs resulting from the moment and p-point classes are smooth enough.
From this work the following questions are open:
\begin{itemize}
\item Is it true that if $f_0(-y)=f_1(y)$ for all $y$, then $u=1/2$ is the minimizer for all/some pairs of uncertainty classes?
\item How should the asymptotic design look like if each random variable $Y_k$ is subjected to different uncertainties, or if $Y_k$ are not mutually independent?
\end{itemize}




\appendices\label{appx}
\section{Existence of a Saddle Value for \eqref{eq45}}\label{appendix2}
\begin{thm}[Application of Sion's minimax theorem \cite{sion}]
A solution to \eqref{eq45} exists if the following conditions hold:
\begin{itemize}
\item The objective function $D_u$ is real valued, upper semi-continuous and quasi-concave on ${\mathscr{G}}_0\times{\mathscr{G}}_1$ for all $u\in[0,1]$.
\item The objective function $D_u$ is lower semi-continuous and quasi-convex on $[0,1]$ for all $(G_0,G_1)\in{\mathscr{G}}_0\times{\mathscr{G}}_1$.
\item $[0,1]$ is a compact convex subset of a linear topological space.
\item ${\mathscr{G}}_0\times{\mathscr{G}}_1$ is a convex subset of a linear topological space.
\end{itemize}
\end{thm}

\begin{IEEEproof}
The objective function is real valued, continuous in $u$ and $(g_0,g_1)$, jointly concave on ${\mathscr{G}}_0\times{\mathscr{G}}_1$ for all $u\in[0,1]$, and convex on $[0,1]$ for all $(G_0,G_1)\in{\mathscr{G}}_0\times{\mathscr{G}}_1$, see $2$nd and $3$rd properties of $D_u$. The set $[0,1]$ is trivially convex and is closed and bounded, hence compact with respect to the standard topology by Heine-Borel theorem \cite[Theorem 2.41]{rudin1976}. Finally, ${\mathscr{G}}_0$ and ${\mathscr{G}}_1$ are convex sets, since $D_f$ is a convex distance. As a result ${\mathscr{G}}_0\times {\mathscr{G}}_1$ is also convex.
\end{IEEEproof}

\section{Proof of Theorem~\ref{theorem01}}\label{appendix3}
\begin{IEEEproof}
For any $t$, for which $I_1(t)>I_0(t)$ we have
\begin{equation}
P_E(n,t)=P_0C_F(n)\exp(-nI_0(t))\,\,\text{as}\,\,n\to\infty,
\end{equation}
since
\begin{align}
\frac{P_E(n,t)}{P_0C_F(n)\exp(-nI_0(t))}&=\frac{P_0C_F(n)\exp(-nI_0(t))+(1-P_0)C_M(n)\exp(-nI_1(t))}{P_0C_F(n)\exp(-nI_0(t))}\nonumber\\
&=1+C(n)\exp(-n(I_1(t)-I_0(t)))\to 1,
\end{align}
where
\begin{equation}
C(n)=\frac{(1-P_0)C_M(n)}{P_0C_F(n)}.
\end{equation}
This argument is true because $I_1(t)-I_0(t)>0$ and $C(n)$ is sub-exponential since from \eqref{eq32} we have
\begin{equation}
\frac{1}{n}(\log((1-P_0)C_M(n))-\log(P_0C_F(n)))=0\,\,\text{as}\,\,n\to\infty.
\end{equation}
Similarly, for the case $I_1(t)<I_0(t)$ we have
\begin{equation}
P_E(n,t)=(1-P_0)C_M(n)\exp(-nI_1(t))\,\,\text{as}\,\,n\to\infty.
\end{equation}
Consequently, as $n\to\infty$ we have
\begin{equation}
P_E(n,t) = \begin{cases} P_0C_F(n)\exp(-nI_0(t)), & I_1(t)>I_0(t)  \\ (1-P_0)C_M(n)\exp(-nI_1(t)), & I_0(t)>I_1(t) \end{cases},
\end{equation}
which can be rewritten as
\begin{equation}
P_E(n,t) = P_0^a(1-P_0)^{1-a}C_F(n)^aC_M(n)^{1-a}\exp(-n\min\{I_0(t),I_1(t)\})
\end{equation}
where $a=\mathbf{1}_{\{I_1>I_0\}}$. Hence, as $n\to\infty$
\begin{align}
\min_t P_E(n,t) &\equiv \min_t \exp(-n\min\{I_0(t),I_1(t)\})\equiv \min_t -n\min\{I_0(t),I_1(t)\}\nonumber\\
&\equiv\max_t \min\{I_0(t),I_1(t)\}\equiv\min_t \max\{I_0(t),I_1(t)\},
\end{align}
since $P_0^a(1-P_0)^{1-a}C_F(n)^aC_M(n)^{1-a}$ is positive and independent of $t$. From \cite[Remark. 5.2.2.]{gulbook}, $I_0$ and $I_1$ are increasing and decreasing functions of $u$, respectively. Let $h_j:u\mapsto t$ be the mapping between maximizing $u$ and $t$ in \eqref{eq28}. It is easy to see that $h_j$ is increasing because it is the derivative of a convex function $\log M_{X_1}^j(u)$ \cite[p. 77]{levy}. Hence, $I_0(t)=I_0(h_0(u))$ and $I_1(t)=I_1(h_1(u))$ are also increasing and decreasing functions respectively, as
\begin{align*}
\frac{d I_0(h_0(u))}{d u}&=I^{'}_0(h_0(u))h_0^{'}(u)\geq 0, \\
\frac{d I_1(h_1(u))}{d u}&=I^{'}_1(h_1(u))h_1^{'}(u)\leq 0.
\end{align*}
Since $I_0$ and $I_1$ are also increasing and decreasing functions of $t$, and furthermore, as $M_{X_1}^1(u)=M_{X_1}^0(u+1)$ for \text{$G_j:=\hat{G}_j$} together with \eqref{eq28} implies \text{$I_1(t)=I_0(t)-t$}, it is true that \text{$I_0(0)=I_1(0)$} and together with \text{$\{t:I_1>I_0\}\equiv\{t:t<0\}$} and \text{$\{t:I_1<I_0\}\equiv\{t:t>0\}$} one can write
\begin{equation}\label{eq35}
I_m(t)=\min\{I_0(t),I_1(t)\}=
\begin{cases}
I_0(t), &  t<0 \\
I_1(t), &  t>0 \\
I_0(0)=I_1(0), &  t=0
\end{cases}.
\end{equation}
Hence, we have
\begin{equation}\label{eq36}
\arg\sup_t I_m(t)=0.
\end{equation}
Notice that we need $G_j:=\hat{G}_j$ in Theorem~\ref{theorem0}. Else, \eqref{eq35} and \eqref{eq36} do not necessarily hold.
\end{IEEEproof}

\bibliographystyle{IEEEtran}
\bibliography{strings4}
\end{document}